\documentclass[pdflatex,sn-basic]{sn-jnl}

\usepackage{xcolor}
\usepackage{bbm}
\usepackage{mathtools}
\jyear{2021}%

\theoremstyle{thmstyleone}
\newtheorem{theorem}{Theorem}
\newtheorem{proposition}[theorem]{Proposition}

\theoremstyle{thmstyletwo}

\theoremstyle{thmstylethree}
\newtheorem{definition}{Definition}

\begin{document}

\title[Sampling via Rejection-Free Partial Neighbor Search]{Sampling via Rejection-Free Partial Neighbor Search}

\author[1]{\fnm{Sigeng} \sur{Chen}}\email{sigeng.chen@mail.utoronto.ca}

\author[1]{\fnm{Jeffrey} \spfx{S.} \sur{Rosenthal}}\email{jeff@math.toronto.edu}

\author[2]{\fnm{Aki} \sur{Dote}}\email{dote.aki@fujitsu.com}

\author[3]{\fnm{Hirotaka} \sur{Tamura}}\email{tamura.hirotaka@dxrlab.com}

\author[4]{\fnm{Ali} \sur{Sheikholeslami}}\email{ali@ece.utoronto.ca}

\affil[1]{\orgdiv{Department of Statistical Sciences}, \orgname{University of Toronto}, \orgaddress{\street{700 University Avenue}, \city{Toronto}, \postcode{M5G 1Z5}, \state{Ontario}, \country{Canada}}}

\affil[2]{\orgname{Fujitsu Ltd.}, \orgaddress{\street{4-1-1 Kamikodanaka Nakahara-ku}, \city{Kawasaki}, \postcode{211-8588}, \state{Kanagawa}, \country{Japan}}}

\affil[3]{\orgname{DXR Laboratory Inc.}, \orgaddress{\street{ 4-38-10 Takata-Nishi, Kohoku-ku}, \city{Yokohama}, \postcode{223-0066}, \state{Kanagawa}, \country{Japan}}}

\affil[4]{\orgdiv{Department of Electrical and Computer Engineering}, \orgname{University of Toronto}, \orgaddress{\street{10 King's College Road}, \city{Toronto}, \postcode{M5S 3G4}, \state{Ontario}, \country{Canada}}}

\date{\today}
  
\abstract{The Metropolis algorithm \citep{metropolis1953equation, hastings1970monte} involves producing a Markov chain to converge to a specified target density $\pi$. In order to improve its efficiency, we can use the Rejection-Free \citep{rosenthal2021jump} version of the Metropolis algorithm, which avoids the inefficiency of rejections by evaluating all neighbors. Rejection-Free can be made more efficient through the use of parallelism hardware. However, for some specialized hardware, such as Digital Annealing Unit \citep{9045100}, the number of units will limit the number of neighbors being considered at each step. Hence, we propose an enhanced version of Rejection-Free known as Partial Neighbor Search, which only considers a portion of the neighbors while using the Rejection-Free technique. This method will be tested on several examples to demonstrate its effectiveness and advantages under different circumstances.}


\maketitle

\section{Introduction}
\label{sec-introduction}

The Monte Carlo method involves the deliberate use of random numbers in a calculation with the structure of a stochastic process \citep{kalos2009monte}. Monte Carlo techniques are based on repeating experiments sufficiently many times to obtain many quantities of interest using the Law of Large Numbers and other statistical inference methods \citep{kroese2014monte}. The three main applications of Monte Carlo methods are optimization, numerical integration, and sampling \citep{kroese2014monte}. This paper focuses on the Markov chain Monte Carlo method for sampling.

The Markov chain Monte Carlo method (MCMC) simulates observations from a target distribution to obtain a chain of states that eventually converges to the target distribution itself. The Metropolis algorithm \citep{metropolis1953equation, hastings1970monte}, an MCMC method, is one of the most popular techniques among its kind \citep{hitchcock2003history}. The Metropolis algorithm produces a Markov chain $\{X_0, X_1, X_2, \dots\}$ on the state space $\mathcal{S}$ and target density function $\pi$, as follows: given the current state $x_k$, the Metropolis algorithm first proposes a new state $y$ from a symmetric proposal distribution $\mathcal{Q}(x_k, \cdot)$; it then accepts the new state (i.e., sets $x_k = y$) with probability $\min(1, \frac{\pi(y)}{\pi(x_k)})$; otherwise, it rejects the proposal (i.e., sets $x_{x + 1} = x_{k}$). This simple algorithm ensures that the Markov chain has $\pi$ as a stationary distribution.

However, the Metropolis algorithm may suffer from the inefficiency of rejections. We have a probability of $\Big[1 - \min(1, \frac{\pi(y)}{\pi(x_k)})\Big]$ to remain at the current state, even though we have spent time proposing a state, computing a ratio of target probabilities, generating a random variable, and deciding not to accept the proposal. Therefore, we proposed the Rejection-Free algorithm in \cite{rosenthal2021jump} to improve the Metropolis algorithm's performance. Furthermore, the parallelism of computer hardware can significantly increase the efficiency of Rejection-Free. The use of parallelism in Rejection-Free combined with simple techniques such as parallel tempering can yield 100x to 10,000x speedups \citep{sheikholeslami2021power}. However, there is a limit to the number of parallel tasks that can be executed simultaneously on most specialized parallelism hardware. Accordingly, Rejection-Free has a ceiling on the number of neighbors that can be evaluated at each step. Consequently, we present an enhanced version of Rejection-Free called Partial Neighbor Search (PNS), which only considers part of the neighbors when applying the Rejection-Free technique, whereas the Rejection-Free technique means considering all selected neighbors and calculating the next state when ignoring any immediately repeated states. PNS has also been developed to solve optimization problems, and it outperforms the Simulated Annealing and Rejection-Free algorithms in many optimization problems, including the QUBO, Knapsack, and 3R3XOR problems; see \cite{chen2022optimization} for further information.

We next review the Metropolis-Hastings algorithm and Rejection-Free algorithm in more detail. Then, in Section~\ref{sec-basicpns}, we introduce our Basic Partial Neighbor Search (Basic PNS) sampling algorithm, which considers subsets of neighbor states for possible moves and calculates the multiplicity list directly from the subsets. Unfortunately, this version of the Markov chain does not converge to the target density. In Section~\ref{sec-unbiasedpns}, we introduce our unbiased version of Partial Neighbor Search (Unbiased PNS), where the sampling distribution will converge to the target density correctly; see Appendix~\ref{appendixA} for the proof. Unlike Rejection-Free, Unbiased PNS can always use the advantage of the parallelism hardware to improve the sampling efficiency, no matter the dimension of the problem. We apply the Unbiased PNS to the QUBO question to illustrate its performance in Section~\ref{sec-qubo}. In addition, we discuss the choice of subsets of the Unbiased PNS for the QUBO question in Section~\ref{sec-chooseparameters}. We further illustrate that we can apply the Unbiased PNS to continuous models in Section~\ref{sec-continuous}. We compare the Metropolis algorithm and Unbiased PNS in a continuous example called the Donuts example to demonstrate the performance of Unbiased PNS in Section~\ref{sec-donuts}. Furthermore, in our optimization paper \citep{chen2022optimization}, the performance of PNS in optimization questions is much better than Rejection-Free and Simulated Annealing. Thus we adapt the Optimization PNS and use it as the burn-in part for sampling in Section~\ref{sec-burninpns}. \cite{geyer2011introduction} stated that burn-in until converging to stationarity is not necessary for MCMC. If we take \citeauthor{geyer2011introduction}'s (\citeyear{geyer2011introduction}) argument, then we can use Optimization PNS to replace the burn-in. On the other hand, we can combine the Optimization PNS and the regular burn-in to get a better algorithm that will converge to stationarity faster. In Appendix \ref{appendixA}, we prove the convergence theorem of our Unbiased PNS algorithm. In addition, in Appendix \ref{appendixC}, we show how to sample proportionally and efficiently on parallelism hardware. Even when we apply the algorithm to a single core implementation, the technique also reduces the time of selecting the next state to some extent.

\subsection{Background on the Metropolis-Hastings algorithm}

Discrete sampling questions usually contain the following essential elements (adapted from the essential elements of Simulated Annealing in \cite{bertsimas1993simulated}):
\begin{enumerate}
    \item a state space $\mathcal{S}$;
    \item a real-valued target distribution $\pi: \mathcal{S} \to [0, 1]$ where $\sum_{x \in \mathcal{S}} \pi(x) = 1$;
    \item $\forall x \in \mathcal{S}$, $\exists$ a proposal distribution $\mathcal{Q}(x, \cdot)$ where $\sum_{y \in \mathcal{S} \backslash \{x\}} \mathcal{Q}(x, y) = 1$, and $\mathcal{Q}(x, y) > 0 \iff \mathcal{Q}(y, x) > 0$, $\forall x, y \in \mathcal{S}$;
    \item $\forall x \in \mathcal{S}$, $\exists$ a neighbor set $\mathcal{N}(x) = \{y \in \mathcal{S} \mid \mathcal{Q}(x, y) > 0\} \subset \mathcal{S} \backslash \{x\}$.
\end{enumerate}
For simplicity, we focus on the discrete cases here. We will talk more about the general state space in Appendix \ref{appendixA}. 

The Metropolis algorithm has been the most successful and influential of all the members of the Monte Carlo method \citep{beichl2000metropolis}. It is designed to generate a Markov chain that converges to a given target distribution $\pi$ on a state space $\mathcal{S}$. The Metropolis-Hastings(M-H) algorithm is a generalized version of the Metropolis algorithm, including the possibility of a non-symmetric proposal distribution $\mathcal{Q}$ \citep{hitchcock2003history}. The M-H algorithm is stated in Algorithm \ref{alg-metropolis}.

\begin{algorithm}
\caption{the Metropolis-Hastings algorithm}\label{alg-metropolis}
\begin{algorithmic}
\State initialize $X_0$
\For{$k$ in $1$ to $K$}
    \State random $Y \in \mathcal{N}(X_{k-1})$ based on $\mathcal{Q}(X_{k-1}, \cdot)$
    \State random $U_k \sim \text{Uniform}(0, 1)$
    \If{$U_k < \frac{\pi(Y) \mathcal{Q}(Y, X_{k-1})}{\pi(X_{k-1}) \mathcal{Q}(X_{k-1}, Y)}$}     
    \State \Comment{accept with probability $\min \Big{\{} 1,  \frac{\pi(Y) \mathcal{Q}(Y, X_{k-1})}{\pi(X_{k-1}) \mathcal{Q}(X_{k-1}, Y)} \Big{\}} $}
    \State $X_{k} \gets Y$ \Comment{accept and move to state $Y$}
    \Else \State $X_{k} \gets X_{k-1}$ \Comment{reject and stay at $X_{k-1}$}
    \EndIf 
\EndFor
\end{algorithmic}
\end{algorithm}

Algorithm \ref{alg-metropolis} ensures the Markov chain $\{X_0, X_1, X_2, \dots, X_K\}$ has $\pi$ as stationary distribution. It follows (assuming irreducibility) that the expected value $E_\pi(h)$ of a functional $h:S\to \mathbb{R}$ with respect to $\pi$ can be estimated by $\frac{1}{K} \sum_{i=1}^K h(X_i)$ for sufficiently large run length $K$. 

In Algorithm \ref{alg-metropolis}, if $U_k \ge \frac{\pi(Y) \mathcal{Q}(Y, X_{k-1})}{\pi(X_{k-1}) \mathcal{Q}(X_{k-1}, Y)}$, then we will remain at the current state, even though we have spent time in proposing a state, computing a ratio of target probabilities, generating a random variable $U_k$, and deciding not to accept the proposal. Such inefficiencies could happen frequently and are considered a necessary evil of the M-H algorithm. Thus, we proposed the Rejection-Free algorithm \citep{rosenthal2021jump} to improve the inefficiency caused by these rejections. 

\subsection{Background on Rejection-Free algorithm for sampling}
\label{subsec-rf}

Before introducing the Rejection-Free algorithm, we must first introduce the jump chain. Given a run $\{X_k\}$ of a Markov chain, we define the jump chain to be $\{J_k, M_k\}$, where $\{J_k\}$ represents the same chain as $\{X_k\}$ except omitting any immediately repeated states, and we use the multiplicity list $\{M_k\}$ to count the number of times the original chain remains at the same state.

For example, if the original chain is 
$$\{X_k\} = \{a, b, b, b, a, a, c, c, c, c, d, d, a, \dots\},$$
then the jump chain and the corresponding multiplicity list would be 
$$\{J_k\}= \{a, b, a, c, d, a, \dots\} \mbox{, } \{M_k\} = \{1, 3, 2, 4, 2, 1, \dots\}.$$

The jump chain itself is also a Markov chain. If we assume the transition probability of the original Markov chain $\{X_k\}$ generated by Algorithm \ref{alg-metropolis} is
\begin{equation}
\begin{aligned}
\label{equa-2}
    P[X_{k} = y \mid X_{k-1} = x] =  \mathcal{Q}(x, y) \min\Bigg\{1, \frac{\pi(y) \mathcal{Q}(y, x)}{\pi(x) \mathcal{Q}(x, y)}\Bigg\},
\end{aligned}
\end{equation} 
Then the transition probabilities $\hat{P}(y \mid x)$ for the jump chain $\{J_k, M_k\}$ is specified by 
\begin{equation}
\begin{aligned}
\label{equa-1}
    \hat{P}(J_k = x \mid J_{k - 1} = x) & = 0\mbox{, } \forall x \in \mathcal{S};  \\
    \\
    \hat{P}(J_k = y \mid J_{k - 1} = x) & = P(X_{k} = y \mid X_{k-1} = x, X_{k-1} \ne x) \\
    & = \frac{P(X_{k} = y \mid X_{k-1} = x)}{\sum_{z \ne x} P(X_{k} = z \mid X_{k-1} = x)} \mbox{, } \forall y \ne x.
\end{aligned}
\end{equation} 
Moreover, the conditional distribution of $\{M_k\}$ given $\{J_k\}$ is equal to the distribution of $1 + G$ where $G$ is a geometric random variable with success probability $1 - P(x \mid x) = \sum_{z \ne x}P(z \mid x)$; see \cite{rosenthal2021jump} for more details.

In addition, for the jump chain $\{J_k, M_k\}_{k=1}^K$, we call the total number of different states $K$ to be the jump sample size, and we call $\sum_{k=1}^K M_k$ to be the original sample size, which is the corresponding length of the original Markov chain. 

Given the above properties of the jump chain, the Rejection-Free algorithm is a sampling method that produces the jump chain as described by Algorithm \ref{alg-rf}. Note that the Rejection-Free algorithm described here can only deal with the discrete cases with at most a finite number of neighbors for all states. We'll review the Rejection-Free for general state space in Section \ref{sec-continuous}.

\begin{algorithm}
\caption{Rejection-Free algorithm for discrete case}\label{alg-rf}
\begin{algorithmic}
\State initialize $J_0$
\For{$k$ in $1$ to $K$}
    \State choose the next jump chain State $J_{k} \in \mathcal{N}(J_{k-1})$ such that $$\hat{P}(J_{k} = y \mid J_{k-1}) \propto \mathcal{Q}(J_{k-1}, y)\min\Bigg\{1, \frac{\pi(y) \mathcal{Q}(y, J_{k-1})}{\pi(J_{k-1}) \mathcal{Q}(J_{k-1}, y)}\Bigg\}$$
    \State calculate multiplicity list $M_{k-1} \gets 1 + G$ where $G \sim \text{Geometric}(p)$ with 
    $$p = \sum_{z \in \mathcal{N}(J_{k-1})}\mathcal{Q}(J_{k-1}, z)\min\Bigg\{1, \frac{\pi(z)\mathcal{Q}(z, J_{k-1})}{\pi(J_{k-1})\mathcal{Q}(J_{k-1}, z)}\Bigg\}$$
\EndFor
\end{algorithmic}
\end{algorithm}
Note that, in Algorithm \ref{alg-rf}, when we need to pick our next state according to the given probabilities, we can use the technique shown in Appendix \ref{appendixC}, which is specially designed for parallelism hardware. In addition, even when the Rejection-Free is applied to a single core implementation, such a technique is still faster than other methods to sample proportionally.

Algorithm \ref{alg-rf} ensures (assuming irreducibility) that the expected value $E_\pi(h)$ of a functional $h: \mathcal{S}\to \mathbb{R}$ with respect to $\pi$ can be estimated by $\frac{\sum_{k=1}^K M_k \, h(J_k)}{\sum_{k=1}^K M_k}$ for sufficiently large run length $K$, while avoiding any rejections. Rejection-Free can lead to great speedup in examples where rejections frequently happen for the M-H algorithm \citep{rosenthal2021jump}.

\section{Basic Partial Neighbor Search algorithm} 
\label{sec-basicpns}

In Algorithm \ref{alg-rf}, we can do this algorithm with parallelism in computer hardware to produce more efficient samples. However, the number of tasks that can be computed simultaneously by the parallelism hardware is not unlimited, while the number of neighbors $\rvert \mathcal{N}(x) \lvert$ can be super large. How can we take full advantage of the Rejection-Free with limited parallel hardware?

Assume the number of neighbors in Rejection-Free is at most N. That is, for $\forall x \in S$, $\lvert \mathcal{N}(x) \rvert \le N$. In addition, assume the number of tasks that can be computed simultaneously by the parallelism hardware is M. If $M > N$, then we can compute the transition probability of the original chain simultaneously by the parallelism hardware, where the transition probability is
\begin{equation}
    P(J_{k} = y \mid J_{k-1}) \propto \mathcal{Q}(J_{k-1}, y)\min\Big\{1, \frac{\pi(y) \mathcal{Q}(y, J_{k-1})}{\pi(J_{k-1}) \mathcal{Q}(J_{k-1}, y)}\Big\}.
\end{equation}
Then the transition probability $\hat{P}$ (defined at Equation \ref{equa-1}) for the Rejection-Free algorithm as stated in Algorithm \ref{alg-rf} is propositional to $P$, $\forall y \ne J_{k-1}$. On the other hand, if $M \le N$, the simplest way to take advantage of parallelism hardware is to evenly distribute the calculation tasks of the transition probabilities to each unit. In this case, each unit of parallelism hardware needs to calculate the probabilities for either $\lfloor \frac{N}{M} \rfloor$ (the floor function) or $\lceil \frac{N}{M} \rceil$ (the ceiling function) times, and then we can put the information from all these parts together for the next step of the algorithm. This method works for processors designed for general purposes, such as Intel and AMD cores. However, these chips are not specially designed for parallel computing, and off-chip communication significantly slows down the transfer rate of data to and from the cores \citep{sodan2010parallelism}. Therefore, using Intel and AMD cores as parallelism hardware is applicable but not ideal.

Moreover, several parallelization hardware specialized for parallel MCMC trials has been proposed. For example, the second generation of Fujitsu Digital Annealer uses a dedicated processor called a Digital Annealing Unit (DAU) \citep{9045100} to achieve high speed. This dedicated processor is designed to minimize communication overhead in arithmetic circuitry and with memory. In addition, the dedicated processor provides a virtually Rejection-Free process, resulting in a throughput that is orders of magnitude faster than that of a general-purpose processor. The problem with this Fujitsu chip is that it is rigidly constrained by on-chip memory capacity relative to the problem size $M$ that can be processed in parallel. For problem sizes $N > M$, it is impossible to compute transition probabilities for all neighborhoods to achieve Rejection-Free or similar parallel trials. The number of neighbors considered in each step must be limited to be within the on-chip memory capacity.

Initially, we want to adapt our Optimization Partial Neighbor Search (Optimization PNS) algorithm from \cite{chen2022optimization} to the sampling question here. Intuitively, we can use the Optimization PNS and add a step for calculating the multiplicity list. The Basic Partial Neighbor Search algorithm (Basic PNS) is shown in Algorithm \ref{alg-basicpns}. Again, we focus on discrete cases with at most a finite number of neighbors here. We will talk about PNS for general state space in Section \ref{sec-continuous} and Appendix \ref{appendixA}.

\begin{algorithm}
\caption{Basic Partial Neighbor Search algorithm}\label{alg-basicpns}
\begin{algorithmic}
\State initialize $J_0$
\For{$k$ in $1$ to $K$}
    \State pick the Partial Neighbor Set $\mathcal{N}_k(J_{k-1}) \subset \mathcal{N}(J_{k-1})$
    \State choose the next jump chain State $J_{k} \in \mathcal{N}_k(J_{k-1})$ such that $$\hat{P}(J_{k} = y \mid J_{k-1}) \propto \mathcal{Q}(J_{k-1}, y)\min\Bigg\{1, \frac{\pi(y) \mathcal{Q}(y, J_{k-1})}{\pi(J_{k-1}) \mathcal{Q}(J_{k-1}, y)}\Bigg\}$$
    \State calculate multiplicity list $M_{k-1} \gets 1 + G$ where $G \sim \text{Geometric}(p)$ with 
    $$p = \sum_{z \in \mathcal{N}_k(J_{k-1})}\mathcal{Q}(J_{k-1}, z)\min\Bigg\{1, \frac{\pi(z) \mathcal{Q}(z, J_{k-1})}{\pi(J_{k-1}) \mathcal{Q}(J_{k-1}, z)}\Bigg\}$$
\EndFor
\end{algorithmic}
\end{algorithm}

The only difference between the Basic PNS (Algorithm \ref{alg-basicpns}) and Rejection-Free (Algorithm \ref{alg-rf}) is that we only calculate the transition probability and all the corresponding values for a subset $\mathcal{N}_k$ of all the neighbors for each step within the loop. Here, $\mathcal{N}_k(J_{k-1})$ is a subset of $\mathcal{N}(J_{k-1})$ at our choice, and the subscript $k$ in $\mathcal{N}_k$ represents the subset of neighbors for step $k$. For example, we can simply say that $\mathcal{N}_k(J_{k-1})$ is a random subset of $\mathcal{N}(J_{k-1})$ with half of the elements. In addition, $\mathcal{Q}_k(X, Y)$ is the corresponding proposal distribution satisfying $\mathcal{Q}_k(x, y) \propto \mathcal{Q}(x, y)$ for $Y \in \mathcal{N}_k(x)$ and $\mathcal{Q}_k(x, y) = 0$ otherwise. However, the Markov chain produced by Algorithm \ref{alg-basicpns} is different from the true MCMC, and it might not converge to the true density $\pi$, as we now show.

\subsection{Example 1 of the Nonconvergence problem by Basic PNS}

\begin{figure}
    \centering
    \includegraphics[width= .7\linewidth]{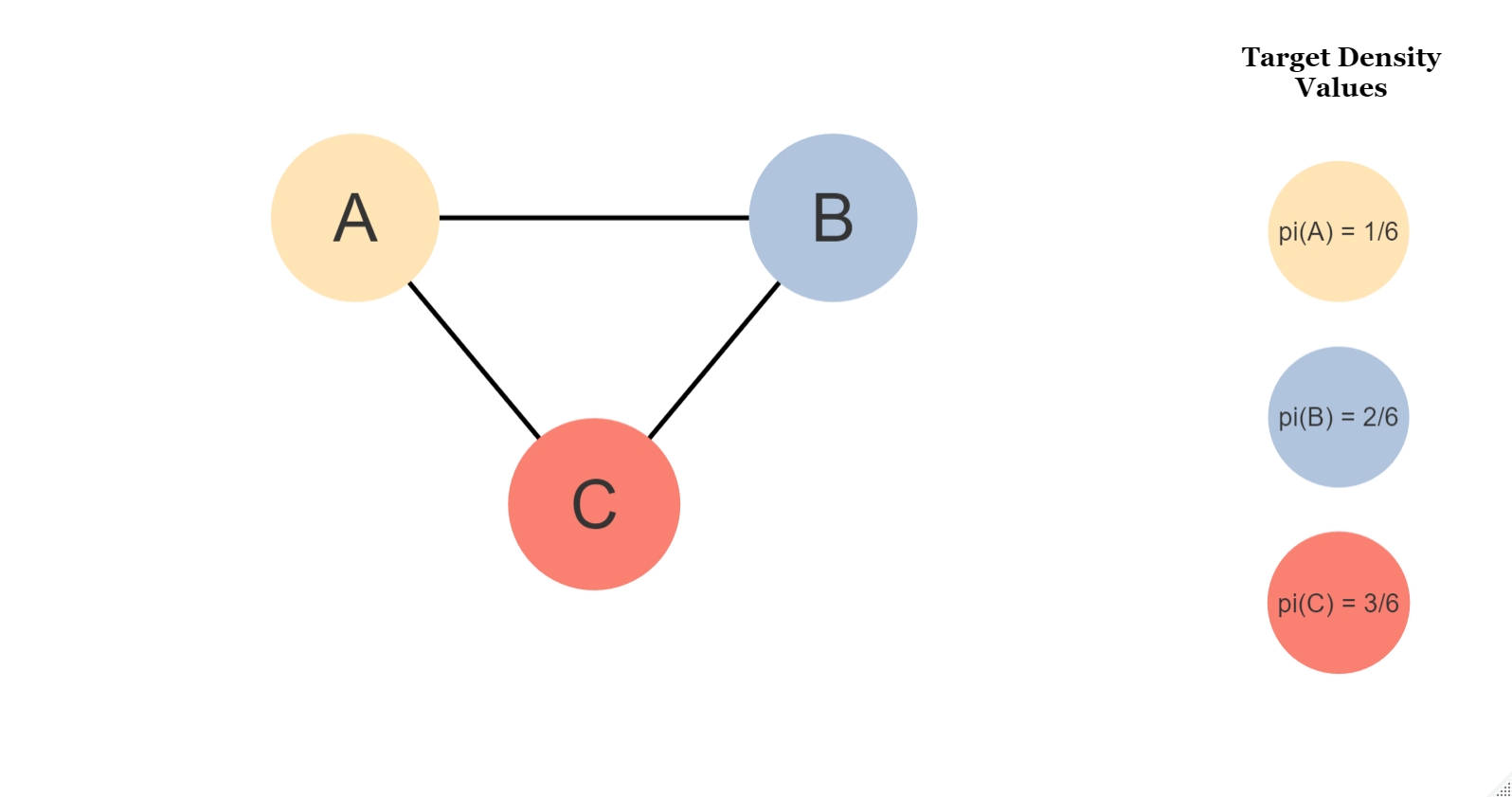}
    \caption{Diagram of Example 1 showing non-convergence property of the Basic PNS.}
    \label{fig-basic001}
\end{figure}

The first example is shown in Figure \ref{fig-basic001}, from which we have $\pi(A) \propto 1$, $\pi(B)  \propto 2$, and $\pi(C) \propto 3$. We consider the Basic PNS algorithm with a uniform proposal distribution $\mathcal{Q}$. In addition, only half of the neighbors are chosen for $\mathcal{N}_k$ at each step. That is, we only need to consider one neighbor each time. 

Then, if the MCMC is located at state A, then $\mathcal{N}(A) = \{B, C\}$. $\mathcal{N}_k(A) = \{B\}$ or $\{C\}$ each with $50\%$ probability, and thus, the algorithm will force the chain to move to either B or C with $50\%$ probability. Similarly, when the Markov chain is located at state B, the next state will be A or C with $50\%$ probability, and when the Markov chain is located at state C, the next state will be A or B with $50\%$ probability.

On the other hand, we can calculate the corresponding multiplicity lists $M_{A}$, $M_{B}$, and $M_{B}$ at state $A$ as follows:
\begin{enumerate}
    \item $\hat{P}[B \mid A] \propto P[B \mid A] = \mathcal{Q}(A, B) \min\{1, \frac{\pi(B)\mathcal{Q}(B, A)}{\pi(A)\mathcal{Q}(A, B)}\} = 0.5$;
    \item $\hat{P}[C \mid A] \propto P[C \mid A] = \mathcal{Q}(A, C) \min\{1, \frac{\pi(C)\mathcal{Q}(C, A)}{\pi(A)\mathcal{Q}(A, C)}\} = 0.5$;
    \item the transition probabilities $\hat{P}$ from A to either B or C in Rejection-Free are both $50\%$;
    \item $M_{A} =  1 + G$ where $G \sim \text{Geom}(P[B \mid A] + P[C \mid A]) = \text{Geom}(1)$
    \item $\mathbb{E}(M_{A}) = 1$
    \item Similarly, we have $\mathbb{E}(M_{B}) = \frac{5}{4}$, $\mathbb{E}(M_{A}) = \frac{9}{4}$
\end{enumerate} 

Thus, for the Basic PNS Chain $\{J_k, M_k\}_{k = 1}^K$ with large $K$, the proportions $\mathcal{P}$ of state A, B, and C in the Markov chain are
\begin{equation}
\begin{aligned}
    \mathcal{P}_{\text{Basic PNS}}(A) & = \frac{\sum_{J_k = A} M_k}{\sum_{k=1}^K M_k} = \frac{1}{1 + \frac{5}{4} + \frac{9}{4}} = \frac{2}{9} \ne \pi(A) = \frac{1}{6}; \\
    \mathcal{P}_{\text{Basic PNS}}(B) & = \frac{5}{18} \ne \pi(B) = \frac{1}{3}; \\
    \mathcal{P}_{\text{Basic PNS}}(C) & = \frac{1}{2} = \pi(C). \mbox{ For state C, it is just a coincidence}
\end{aligned}
\end{equation}
This example shows that the samples from Basic PNS are not converging to the target density $\pi$. 

\subsection{Example 2 of the Nonconvergence problem by Basic PNS}

\begin{figure}
    \centering
    \includegraphics[width= .7\linewidth]{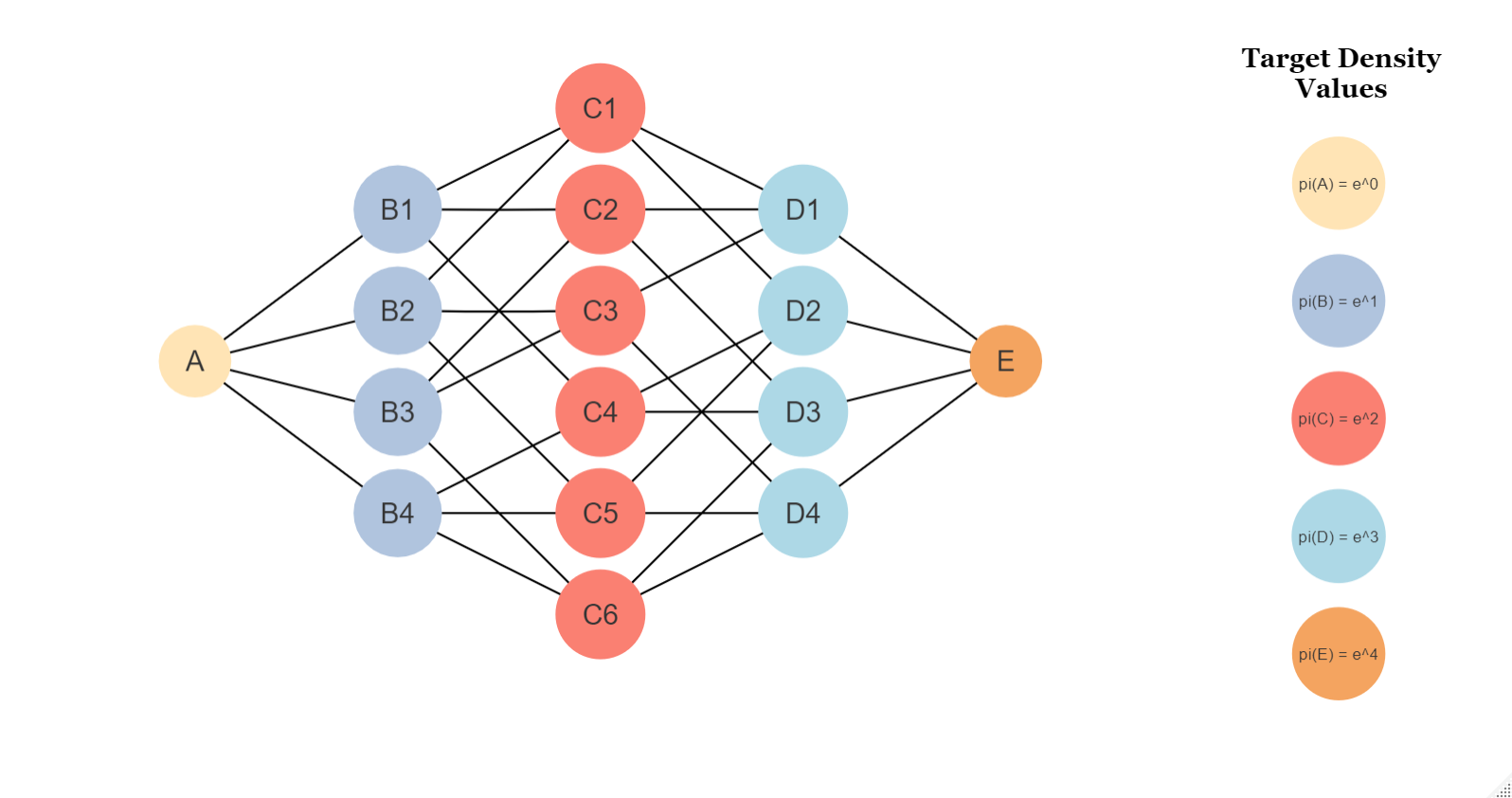}
    \caption{Diagram of Example 2 showing non-convergence property of the Basic PNS.}
    \label{fig-basic002}
\end{figure}

The second example is shown in Figure \ref{fig-basic002}, which is a much larger problem compared to the first example. We have 16 states in example 2. All States are connected to exactly four states. The target density is described as $\pi(A) \propto 1$, $\pi(B_1)  =\pi(B_2)  = \pi(B_3)  = \pi(B_4) \propto e$, $\pi(C_1)  =\pi(C_2)  =  \dots = \pi(C_6) \propto e^2$, $\pi(D_1)  =\pi(D_2) = \pi(D_3) = \pi(D_4) \propto e^3$, and $\pi(E) \propto e^4$. This example is too large to be calculated by hand, so we use simulations to calculate the limiting distribution of the samples. The convergence of the sampling distribution is measured by the Total Variation Distance (TVD).

Given the Markov chain $\{X_k\}_{k=1}^K$ generated by Metropolis algorithm, the sampling distribution is defined as $\mathcal{P}_{\text{Sampled}}(x) = \frac{\sum_{k=1}^K \mathbbm{1}(X_k = x)}{K}$, $\forall x \in \mathcal{S}$, where $\mathbbm{1}$ represents the indicator function. In addition, for the jump chain $\{J_k, M_k\}_{k=1}^K$ generated by either Rejection-Free or PNS, the sampling distribution is defined as $\mathcal{P}_{\text{Sampled}}(x) = \frac{\sum_{k=1}^K M_k \times \mathbbm{1}(J_k = x)}{\sum_{k=1}^K M_k}$, $\forall x \in \mathcal{S}$. The corresponding TVD values in both cases are defined as
\begin{equation}
    \text{TVD}(\mathcal{P}_{\text{Sampled}}, \pi)= \frac{1}{2} \sum_{x \in \mathcal{S}} \Big\lvert \mathcal{P}_{\text{Sampled}}(x) - \pi(x) \Big\rvert.
\end{equation}
According to the definition, TVD is strictly between $[0, 1]$. When the sampling distribution $\mathcal{P}_{\text{Sampled}}$ gets closer to the target distribution $\pi$, TVD will decrease to 0. In other words, convergence to stationarity is described by how quickly TVD decreases to 0.

\begin{figure}
    \centering
    \includegraphics[width= \linewidth]{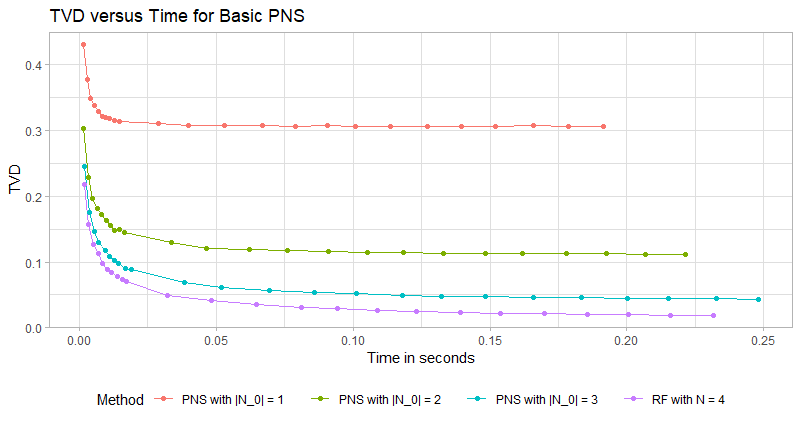}
    \caption{Average values of TVD between samples and the target density $\pi$ for Example 2 as a function of average CPU time in seconds for four scenarios: Rejection-Free and Basic PNS with three different Partial Neighbor Set sizes. Each dot within the plot represents the result of the average TVD value and average CPU time in seconds from 1000 simulation runs given a certain original sample size, where the sizes are $\{50, 100, 150, 200, \dots, 500, 1000, 1500, 2000, \dots, 7500\}$.}
    \label{fig-basic003}
\end{figure}
    
The simulation results are shown in Figure \ref{fig-basic003}. For a given amount of samples ($K = 50$, $100$, $150$, $200$, $\dots$, $500$, $1000$, $1500$, $2000$, $\dots$, $7500$), we did $1000$ simulations for each of them. The TVD values and the times here are the average values from these 1000 simulations. We compared four methods: Rejection Free and Basic PNS with three different subset sizes. The Markov chains, produced by Rejection-Free, will converge to the target density, so the TVD value gets close to $0$ at last. For PNS with $\lvert \mathcal{N}_0 \rvert = 1$, we select one random neighbor among all four neighbors at a time, forcing the chain to move to that state. This method is the worst, and it converges at around 0.3. PNS with $\lvert \mathcal{N}_0 \rvert = 2$ means that we randomly select two neighbors at every step and apply the Rejection-Free technique (select from these two states by probability proportional to the transition probability, and calculate the multiplicity list by the average of the transition probabilities). All three PNS algorithms are not converging to the target density $\pi$. 

Both Examples show that the samples from Basic PNS will not converge to the target distribution $\pi$. Thus, we turn attention to a more promising avenue, the unbiased version of the Partial Neighbor Search algorithm, where convergence to stationarity is guaranteed.

\section{Unbiased Partial Neighbor Search algorithm}
\label{sec-unbiasedpns}

We first need to review the alternating chains technique for the Rejection-Free algorithm \citep{rosenthal2021jump} to introduce our upgraded algorithm version.

\subsection{Alternating Chains for Rejection-Free}

We may wish to switch between two or more different proposal distributions for the M-H algorithm. An example of the M-H algorithm with alternating chains for every $L_0$ steps among $\mathcal{I}$ proposal distributions $\mathcal{Q}^0, \mathcal{Q}^1, \dots, \mathcal{Q}^{\mathcal{I}-1}$ is shown as Algorithm \ref{alg-mhac}. 

\begin{algorithm}
\caption{Metropolis-Hasting algorithm with Alternating Chains}\label{alg-mhac}
\begin{algorithmic}
\State initialize $i \gets 0$ \Comment{start with proposal distribution $\mathcal{Q}^0$}
\State initialize $L \gets L_0$  \Comment{start with $L_0$ remaining samples}
\State initialize $X_0$
\For{$k$ in $1$ to $K$}
    \State random $Y$ based on $\mathcal{Q}^i(X_{k-1}, \cdot)$ 
    \State random $U_k \sim \text{Uniform}(0, 1)$
    \If{$U_k < \frac{\pi(Y) \mathcal{Q}^i(Y, X_{k-1})}{\pi(X_{k-1}) \mathcal{Q}^i(X_{k-1}, Y)}$} \State \Comment{accept with probability $\min \Big{\{} 1,  \frac{\pi(Y) \mathcal{Q}^i(Y, X_{k-1})}{\pi(X_{k-1}) \mathcal{Q}^i(X_{k-1}, Y)} \Big{\}} $}
    \State $X_{k} \gets Y_k$  \Comment{accept and move to state $Y$}
    \Else \State $X_{k} \gets X_{k-1}$ \Comment{reject and stay at $X_{k-1}$}
    \EndIf 
    \State $L \gets L - 1$ \Comment{one less remaining sample from the proposal distribution}
    \If{$L = 0$} \Comment{if we don't have enough remaining samples}
        \State $i \gets {i + 1 \mod \mathcal{I}}$ \Comment{switch to the next proposal distribution}
        \State $L \gets L_0$ \Comment{$L_0$ remaining states for the next proposal distribution}
    \EndIf
\EndFor
\end{algorithmic}
\end{algorithm}

However, if we proceed with alternating chains naively for Rejection-Free, it can lead to bias. For each proposal distribution $\mathcal{Q}_i$,  we need to get the same amount of samples by the original sample size ($\sum_{k=1}^K M_k$) instead of the jump sample size ($K$) to fix the bias problem. For $\mathcal{I}$ proposal distributions $\mathcal{Q}_0, \mathcal{Q}^1, \dots, \mathcal{Q}^{\mathcal{I}-1}$, the corresponding neighbor sets are $\mathcal{N}^0, \mathcal{N}^1, \dots, \mathcal{N}^{\mathcal{I}-1}$ where $\mathcal{N}^i(x) = \{y: y \in \mathcal{S}, \mathcal{Q}^i(x, y) > 0\}$ for $i = 0, 1, \dots, {\mathcal{I}-1}$. Then, if we choose to switch between proposal distributions for $L_0$ original samples, we can do alternating chains in a Rejection-Free manner as Algorithm \ref{alg-rfac}.

\begin{algorithm}
\caption{Rejection Free algorithm with Alternating Chains}\label{alg-rfac}
\begin{algorithmic}
\State initialize $i \gets 0$ \Comment{start with proposal distribution $\mathcal{Q}^0$}
\State initialize $L \gets L_0$  \Comment{start with $L_0$ remaining original samples}
\State initialize $J_0$
\For{$k$ in $1$ to $K$}
    \State calculate multiplicity list $m \gets 1 + G$ where $G \sim \text{Geometric}(p)$ with 
    $$p = \sum_{z \in \mathcal{N}^i(J_{k-1})}\mathcal{Q}^i(J_{k-1}, z)\min\Bigg\{1, \frac{\pi(z)\mathcal{Q}^i(z, J_{k-1})}{\pi(J_{k-1})\mathcal{Q}^i(J_{k-1}, z)}\Bigg\}$$
    \If{$m \le L$} \Comment{if we have enough remaining original samples}
        \State  $M_{k-1}  \gets m$, $L \gets L - m$
        \State choose the next jump chain State $J_{k} \in \mathcal{N}^i(J_{k-1})$ such that 
        $$\hat{P}(J_{k} = y \mid J_{k-1}) \propto \mathcal{Q}^i(J_{k-1}, y)\min\Bigg\{1, \frac{\pi(y) \mathcal{Q}^i(y, J_{k-1})}{\pi(J_{k-1}) \mathcal{Q}^i(J_{k-1}, y)}\Bigg\}$$
    \Else \Comment{if we don't have enough remaining original samples}
        \State $M_{k-1}  \gets L$, $L \gets L_0$, $J_{k} \gets J_{k-1}$, $i \gets (i + 1 \mod \mathcal{I})$ \State \Comment{stay at $J_{k-1}$ for $L$ times and switch to the next $\mathcal{N}^i$}
    \EndIf
\EndFor
\end{algorithmic}
\end{algorithm}

Algorithm \ref{alg-rfac} is equivalent to Algorithm \ref{alg-mhac} except that algorithm \ref{alg-rfac} computes immediate repeated state for each proposal distribution all at once. As such, it has no bias, is consistent, and will converge to the target distribution correctly.

\subsection{Alternating Chains for Partial Neighbor Search}
\label{subsec-alternating}

Alternating Chains can also be applied to PNS. We first define the meaning of Partial Neighbor Sets here. For simplicity, we focus on discrete cases here and will define the Partial Neighbor Sets for general state space in Appendix \ref{appendixA}.

Before we start our Markov chain, we have a proposal distribution $\mathcal{Q}$ with a corresponding neighbor set $\mathcal{N}$ where $\mathcal{N}(x) \coloneqq \{y \in \mathcal{S} \mid \mathcal{Q}(x, y) > 0\}$. A Partial Neighbor Set means any function $\mathcal{N}_i$ satisfies the following conditions:
\begin{enumerate}
    \item $\mathcal{N}_i: \mathcal{S} \to \mathbf{P}(\mathcal{S})$, where $\mathcal{S}$ is the state space, and $\mathbf{P}(\mathcal{S})$ is the power set of $\mathcal{S}$;
    \item $\mathcal{N}_i(x) \subset \mathcal{N}(x)$, $\forall x \in \mathcal{S}$;
    \item $y \in \mathcal{N}_i(x)  \iff x \in \mathcal{N}_i(y)$, $\forall x, y \in \mathcal{S}$;
\end{enumerate}
Usually, we want to pick $\mathcal{N}_i$ such that $\lvert \mathcal{N}_i(x) \rvert < \lvert \mathcal{N}(x) \rvert$ to perform proper PNS. In addition, to insure irreducibility, we need to make sure $\cup_{i=0}^{\mathcal{I}-1} \mathcal{N}_i(x)= \mathcal{N}(x)$ for all $x \in \mathcal{S}$. The corresponding proposal distribution is defined to be $\mathcal{Q}_i(x, y): \mathcal{S} \times \mathcal{S} \to \mathbb{R}$, where $\mathcal{Q}_i(x, y) \propto \mathcal{Q}(x, y)$ for $y \in \mathcal{N}_i(x)$ and $\mathcal{Q}_i(x, y) = 0$ otherwise;

Therefore, we propose the Unbiased Partial Neighbor Search (Unbiased PNS) with Alternating Chains for every $L_0$ original samples as shown in Algorithm \ref{alg-unbiasedpns}. The proof that the Markov chain produced by Unbiased PNS will converge to the target distribution $\pi$ is shown in Appendix \ref{appendixA}. 

Again, in Algorithm \ref{alg-unbiasedpns}, when we need to pick our next state according to the given probabilities, we can use the technique shown in Appendix \ref{appendixC}, which is faster than other methods to sample proportionally.

\begin{algorithm}
\caption{Unbiased Partial Neighbor Search}\label{alg-unbiasedpns}
\begin{algorithmic}
\State select $\mathcal{N}_i$ for $i = 0, 1, \dots, \mathcal{I}-1$ where $\cup_{i=0}^{\mathcal{I}-1} \mathcal{N}_i(X)= \mathcal{N}(X)$
\State initialize $i \gets 0$ \Comment{start with proposal distribution $\mathcal{Q}_0$}
\State initialize $L \gets L_0$  \Comment{start with $L_0$ remaining original samples}
\State initialize $J_0$
\For{$k$ in $1$ to $K$}
    \State calculate multiplicity list $m \gets 1 + G$ where $G \sim \text{Geometric}(p)$ with 
    $$p =  \sum_{z \in \mathcal{N}_i(J_{k-1})}\mathcal{Q}_i(J_{k-1}, z)\min\Bigg\{1, \frac{\pi(z)\mathcal{Q}_i(z, J_{k-1})}{\pi(J_{k-1})\mathcal{Q}_i(J_{k-1}, z)}\Bigg\}$$
    \If{$m \le L$} \Comment{if we have enough remaining original samples}
        \State  $M_{k-1}  \gets m$, $L \gets L - m$
        \State choose the next jump chain State $J_{k} \in \mathcal{N}_i(J_{k-1})$ such that $$\hat{P}(J_{k} = y \mid J_{k-1}) \propto \mathcal{Q}_i(J_{k-1}, y)\min\Bigg\{1, \frac{\pi(y) \mathcal{Q}_i(y, J_{k-1})}{\pi(J_{k-1}) \mathcal{Q}_i(y, J_{k-1})}\Bigg\}$$
    \Else \Comment{if we don't have enough remaining original samples}
        \State $M_{k-1}  \gets L$, $L \gets L_0$, $J_{k} \gets J_{k-1}$, $i \gets (i + 1 \mod \mathcal{I})$ \State \Comment{stay at $J_{k-1}$ for $L$ times and switch to the next $\mathcal{N}_i$}
    \EndIf
\EndFor
\end{algorithmic}
\end{algorithm}

The Markov chains produced by Algorithm \ref{alg-unbiasedpns} will converge to the target distribution, but how is its efficiency compared to the Metropolis-Hasting algorithm and Rejection-Free? We will compare these three algorithms with some simulations in Section \ref{sec-qubo}.

\section{Application to QUBO model}
\label{sec-qubo}

Quadratic unconstrained binary optimization (QUBO) has been rising in importance in combinatorial optimization because of its wide range of applications in finance and economics to machine learning \citep{kochenberger2014unconstrained}. It can also be used as a sampling question, which aims to sample from the distribution
\begin{equation}
    \pi(x) = \exp\{x^T Q x\} \mbox{, where } x \in \{0, 1\}^N
\end{equation}
for a given $N$ by $N$ matrix $Q$ (usually symmetric or upper triangular).

To run our algorithm, we used uniform proposal distributions among all neighbors where the neighbors are defined as binary vectors with Hamming distance 1. That is, $\mathcal{Q}(x, y) = \frac{1}{N}$ for $\forall y$ such that $\lvert x - y \rvert = \sum_{i=1}^N \lvert x_i - y_i \rvert = 1$, $\forall x, y \in \{0, 1\}^N$. Thus, the neighbors are all binary vectors different by one flip. For the first simulation here, the PNS neighbor sets $\mathcal{N}_0$, $\mathcal{N}_1$ are chosen systematically, where $\mathcal{N}_0$ represents flip entries from 1 to $\lfloor \frac{N}{2} \rfloor$, and $\mathcal{N}_1$ represents flip entries from $\lfloor \frac{N}{2} \rfloor + 1$ to $N$. We will discuss many other choices for the PNS neighbor sets in Section \ref{sec-chooseparameters}

\begin{figure}
    \centering
    \includegraphics[width= \linewidth]{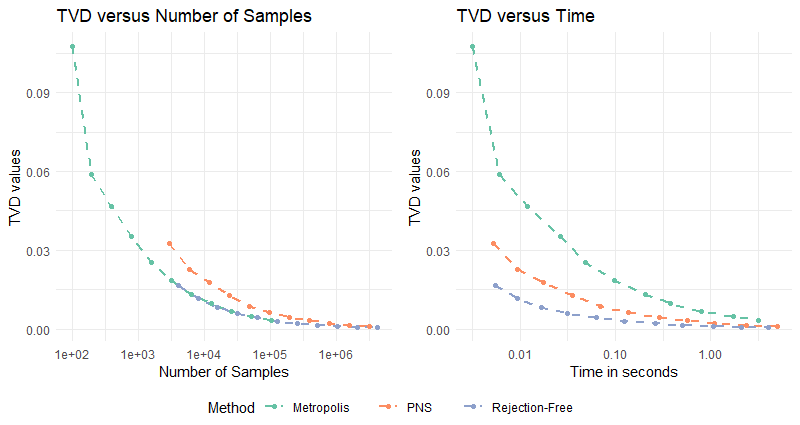}
    \caption{Average values of TVD between sampling and target density $\pi$ as a function of the number of iterations (left) and average time in seconds (right) for three methods: Metropolis algorithm, Rejection-Free, and Unbiased PNS. We used upper triangular $16 \times 16$ QUBO matrix, generated randomly by $Q_{i, j} \sim N(0, 10^2)$ for upper triangular elements. Each dot within the plot represents the result of the average TVD value and time used for 1000 simulation runs given certain original sample sizes. The original sample sizes for the Metropolis algorithm are $\{300, 600, 1200, 2400, \dots, 3072000\}$. The original samples from Rejection-Free are 40x more than those from Metropolis, and the original samples from Unbiased PNS are 30x more than those from Metropolis. We choose these sizes to get a close average CPU time for all three methods. For Unbiased PNS, we used $\lvert \mathcal{N}_k \rvert = 8$ and $L_0 = 100$.}
    \label{fig-unbiased001}
\end{figure}

Figure \ref{fig-unbiased001} shows the results for comparing the Metropolis algorithm, Rejection-Free, and Unbiased PNS by sampling from a $16\times16$ QUBO question to a single-core implementation. The QUBO matrix $Q$ is an upper triangular matrix, where the non-zero elements were generated randomly by $Q_{i, j} \sim \text{Normal}(0, 10^2)\mbox{, } \forall i \le j$. We compare the TVD values from the Metropolis Algorithm, Rejection-Free, and Unbiased PNS with different original sample sizes. For the Metropolis algorithm, the numbers of original samples are $\{100, 200, 400, 800, \dots, 102400\}$. The number of original samples for Rejection-Free is 40 times more than the number for the Metropolis algorithm, and the number for Unbiased PNS is 30 times more. We used these many numbers of original samples to make the run-time for all three algorithms to be about the same. For each given number of original sample sizes, we simulated 1000 runs, recorded the corresponding TVD values and times used for the sampling part, and calculated the average values given the number of original samples. Note that the average time represents the CPU time for where the algorithm is calculated by running the algorithm on a single-core implementation. In addition, before we generate the samples, we apply the algorithm for the same number of steps for burn-in.

From Figure \ref{fig-unbiased001}, we can see that the quality of the samples by the Metropolis algorithm and Rejection-Free are the same given the original sample sizes. This result is consistent with our conclusion that Rejection-Free is identical to the Metropolis algorithm, except Rejection-Free generates the same states simultaneously with all immediately repeated states. Thus, these two algorithms are different only by the CPU time. In addition, the quality of the samples by Unbiased PNS is worse than both the Metropolis algorithm and Rejection-Free given a certain number of original samples because each Partial Neighbor Set is biased within its $L_0$ original samples, while the combination of them is unbiased. Thus, the average TVD value for Unbiased PNS is more significant for the same amount of original samples. However, for a given amount of CPU time, the performance of Unbiased PNS is much better than the Metropolis algorithm and worse than Rejection-Free.

In this case, Unbiased PNS can provide significant speedups compared to the Metropolis algorithm. On the other hand, we did not expect the Unbiased PNS can beat Rejection-Free under this circumstance. Unbiased PNS is worse than Rejection-Free in two aspects. First, the Unbiased PNS is biased within each $L_0$ original samples. In addition, at the end of each $L_0$ original samples, the algorithm is very likely to reject once and stay in the same state. Thus, Unbiased PNS is not entirely rejection-free anymore and usually rejects once for every $L_0$ original samples. 

However, we need the Unbiased PNS because we may not have as many circuit blocks in the parallelism hardware as we want. Thus, we can, at most, consider a limited number of neighbors for some specialized hardware, such as DA. Thus, Rejection-Free is not applicable in this case, and we would need the help of Unbiased PNS, which is better than applying the Metropolis algorithm.

Again, parallelism in computer hardware can increase the speed for both Rejection-Free and Unbiased PNS by mapping the calculation of the transition probabilities for different neighbors onto different cores \citep{rosenthal2021jump}. Besides that, we can also use multiple replicas at different temperatures, such as in parallel tempering, or deploy a population of replicas at the same temperature \citep{sheikholeslami2021power}. Combining these methods by parallelism can yield 100x to 10,000x speedups for both Rejection-Free and Unbiased PNS \citep{sheikholeslami2021power}.

\section{Optimal Choice for the Partial Neighbors} 
\label{sec-chooseparameters}

In the section \ref{sec-qubo}, we used two systematically pre-selected neighbor sets $\mathcal{N}_0$, $\mathcal{N}_1$. However, for the optimization version of the QUBO question, we concluded that random Partial Neighbor Sets are better than systematic Partial Neighbor Sets; see \cite{chen2022optimization}. Thus, we compare two ways of choosing partial neighbor sets here: systematic and random. For simplicity, assume that we have $N$ neighbors for all states, and we use Unbiased PNS neighbor sets of size $n$. Therefore, we have $\binom{N}{n}$ different partial neighbor sets. For systematic method, we choose $\mathcal{I}$ PNS neighbor sets $\{\mathcal{N}_i\}_{i=1}^\mathcal{I}$, where $\cup_{i=1}^\mathcal{I} \mathcal{N}_i(x) = \mathcal{N}(x)$. We proceed with each Partial Neighbor Sets within the loop for $L_0$ original samples. We use the notation $\mathcal{N}_i(x)$ for systematic Partial Neighbor Sets because $\mathcal{N}_i(x)$ is pre-determined for $i = 1, 2, \dots, \mathcal{I}$. On the other hand, for random Partial Neighbor Sets, we choose a new set $\mathcal{N}_k$ from all $\binom{N}{n}$ potential Partial Neighbor Sets after each $L_0$ original samples. We use the notation $\mathcal{N}_k(x)$ for random partial neighbor sets because $\mathcal{N}_k(x)$ can be different for every PNS step, and the subscript $k$ represents the special partial neighbor set for step k. For both methods, $\mathcal{Q}_i(x, y), \mathcal{Q}_k(x, y) \propto \mathcal{Q}(x, y)$ for $y \in \mathcal{N}_i(x)$, and $\mathcal{Q}_i(x, y) = \mathcal{Q}_k(x, y) = 0$ otherwise. 

To compare the above two methods of selecting Partial Neighbor Sets, we apply them to the previous $16 \times 16$ QUBO question, and we test the following four scenarios: 
\begin{itemize}
    \item two systematic partial neighbor sets where the first set considers flipping the first half of the bits, and the second set considers flipping the second half of the bits;
    \item four systematic partial neighbor sets where each set  considers flipping a quarter of the bits;
    \item random partial neighbor sets with $\frac{N}{2}$ partial neighbors; that is, each set considers flipping a random set of bits with size $\frac{N}{2}$;
    \item random partial neighbor sets with $\frac{N}{4}$ partial neighbors; that is, each set considers flipping a random set of bits with size $\frac{N}{4}$;
\end{itemize}   

\begin{figure}
    \centering
    \includegraphics[width= \linewidth]{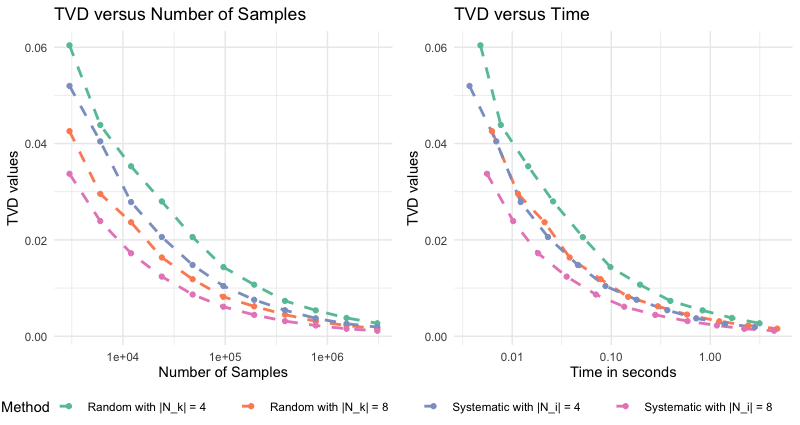}
    \caption{Average values of TVD between sampling and target density $\pi$ as a function of the number of iterations (left) and average time in seconds (right) for four scenarios: Systematic PNS and Random PNS, each with Partial Neighbor Set sizes of 4 and 8. Random upper triangular $16 \times 16$ QUBO matrix is generated randomly by $Q_{i, j} \sim N(0, 10^2)$ for upper triangular elements. Each dot within the plot represents the average TVD value and time used for 1000 simulation runs given a certain original sample size, where the sizes are $\{3000, 6000, 12000, 24000, \dots, 3072000\}$. For all PNS, we used $L_0 = 100$.}
    \label{fig-choiceMethod001}
\end{figure}

The result is shown in Figure \ref{fig-choiceMethod001}; for this case, systematic Partial Neighbor Sets are better than random Partial Neighbor Sets. However, random Partial Neighbor Sets can be better when we run the same code with a different random seed. After running this simulation for 100 different random seeds, the systematic neighbor sets are better 56 times. Thus, we conclude that the performance of these two Partial Neighbor Sets is close to each other. We will continue using the systematic Partial Neighbor Sets in our later simulation.

In previous simulations, we naively use we used $\lvert \mathcal{N}_i(x) \rvert = 4$ or $8$ and $L_0 = 100$ in previous examples. What is the optimal choice for $\lvert \mathcal{N}_i(x) \rvert$? We want to compare $\lvert \mathcal{N}_i(x) \rvert = 2, 4, 6, 8, \dots, 14$ by the QUBO question. In previous cases, we only used the systematic Partial Neighbor Set size $n$ that can be divided evenly by $N$. For other $n$ such as $14$, we used the following Partial Neighbor Sets: first we consider flipping bits $1$ to $14$, then we consider bits $15$, $16$, and $1$ to $12$, then $13$, $14$, $15$, $16$, and $1$ to $10$, etc. We used eight systematic Partial Neighbor Sets of size $14$.

\begin{figure}
    \centering
    \includegraphics[width= \linewidth]{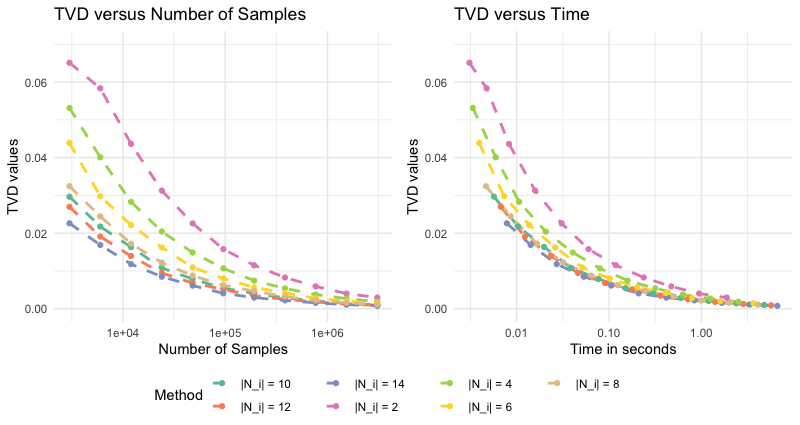}
    \caption{Average values of TVD between sampling and target density $\pi$ as a function of the number of iterations (left) and average time in seconds (right) for four scenarios: Unbiased PNS with different partial neighbor set sizes \{2, 4, 6, \dots, 14\}. Random upper triangular $16 \times 16$ QUBO matrix is generated randomly by $Q_{i, j} \sim N(0, 10^2)$ for upper triangular elements. Each dot within the plot represents the average TVD value and time used for 1000 simulation runs given a certain original sample size, where the sizes are $\{3000, 6000, 12000, 24000, \dots, 3072000\}$. For all PNS, we used $L_0 = 100$.}
    \label{fig-choiceM001}
\end{figure}

Figure \ref{fig-choiceM001} shows the results for comparing the Unbiased PNS with $\lvert \mathcal{N}_i\rvert \equiv 2$, $4$, $6$, $8$, $10$, $12$, and $14$ for $\forall X \in \{0, 1\}^{16}$. Every other simulation setting is the same as the previous simulations for the QUBO question. The choice of $L_0$ is still 100. According to the left plot, we can say that given the same amount of original samples, the Markov chain from $\lvert \mathcal{N}_i\rvert = 14$ is the least biased. On the other hand, from the right plot, we can conclude that, given the same amount of CPU time, the sample quality from $\lvert \mathcal{N}_i\rvert = 14$ is the best. In addition, the performances are close to each other for all cases where $\lvert \mathcal{N}_i\rvert \ge 8$. Note that a single-core implementation makes all these comparisons by the CPU time, and parallelism hardware can provide speedups. Intuitively, the more tasks that can be calculated simultaneously, the greater the speedup. Thus, if we apply our Unbiased PNS on parallelism hardware with a limited number of parallel tasks that can be computed simultaneously, we should choose the largest possible partial neighbor set size $\lvert \mathcal{N}_i\rvert$.

\begin{figure}
    \centering
    \includegraphics[width= \linewidth]{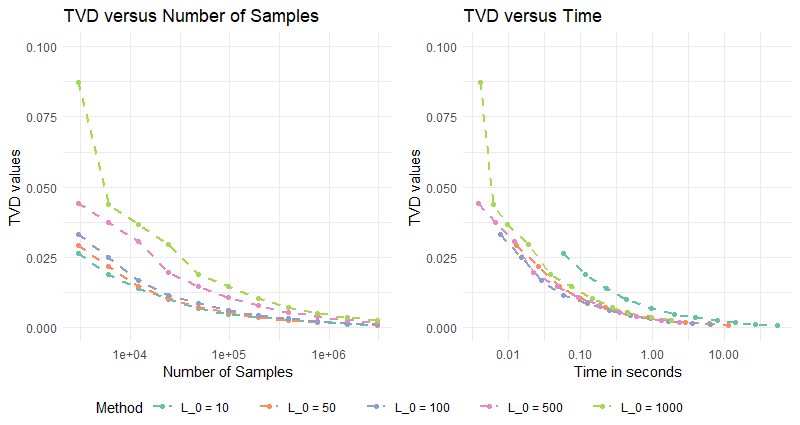}
    \caption{Average values of TVD between sampling and target density $\pi$ as a function of the number of iterations (left) and average time in seconds (right) for Unbiased Partial Neighbor Search with different sizes of $L_0$. Random upper triangular $16 \times 16$ QUBO matrix is generated randomly by $Q_{i, j} \sim N(0, 10^2)$ for upper triangular elements. Each dot within the plot represents the average TVD value and time used for 1000 simulation runs given a certain original sample size, where the sizes are $\{300, 600, 1200, 2400, \dots, 3072000\}$. For all PNS, we used $\lvert \mathcal{N}_i \rvert = 8$.}
    \label{fig-choiceL001}
\end{figure}

Furthermore, Figure \ref{fig-choiceL001} shows the results for comparing the Unbiased PNS with $L_0 = 10$, $50$, $100$, $500$, and $1000$. Again, every other settings of the simulation is the same, and $\lvert \mathcal{N}_i\rvert$ is still 8, $\forall x \in \{0, 1\}^{16}$. The left plot shows that given the same samples, the Markov chain from $L_0 = 10$ is the least biased. However, the right plot shows that, given the same amount of CPU time, the TVD values are about the same except $L_0 = 10$. The case with $L_0 = 100$ is slightly better than the other cases, but the difference is not too large. $L_0 = 10$ becomes the worst since such $L_0$ has too many rejections (about one rejection for every ten samples). Thus, in practice, the choice of $L_0$ is not that important as long as it is not extreme.

\section{Continuous Models}
\label{sec-continuous}

We talked about the application of Unbiased PNS to discrete cases in the previous sections. Can we apply Unbiased PNS on continuous models? We first review how to apply Rejection-Free on general (continuous) state space as in Theorem 13 from \cite{rosenthal2021jump}.

Let $\mathcal{S}$ be a general state case, and $\mu$ a $\sigma$-finite reference measure on $\mathcal{S}$. Suppose a Markov chain on $\mathcal{S}$ has transition probabilities $P(x, d y) \propto q(x, y) \mu(d y)$ for $q: \mathcal{S} \times \mathcal{S} \to [0, 1]$. Again let $\hat{P}$ be the transitions for the corresponding jump chain ${J_k}$ with multiplicities ${M_k}$. Then:

\begin{enumerate}
    \item $\hat{P}(x, \{x\}) = 0$, and for $x \ne y$, $\hat{P}(x, d y) = \frac{q(x, y)}{\int q(x, z) \mu(d z)} \mu(d y)$
    \item  The conditional distribution of $M_k$ given $J_k$ is equal to the distribution of $1+G$ where $G$ is a geometric random variable with success probability $p = \alpha(J_k)$ where $\alpha(x) = P[X_{k+1} \ne x \mid X_{k} = x] = \int q(x, z) \mu(d z) = 1 - r(x) = 1 - P(x \mid x)$
    \item If the original chain is $\phi$-irreducible (see, e.g., \cite{meyn2012markov}) for some positive $\sigma$-finite measure $\phi$ on $\mathcal{X}$, then the jump chain is also $\phi$-irreducible for the same $\phi$.
    \item  If the original chain has stationary distribution $\pi(x) \mu(d x)$, then the jump chain has stationary distribution given by $\hat{\pi}(x) = c \alpha(x) \pi(x) \mu(d x)$ where $c^{-1} = \int \alpha(y) \pi(y) \mu(d y)$
    \item If $h : \mathcal{S} \to \mathbb{R}$ has finite expectation, then with probability $1$,
    $$\lim_{K \to \infty} \frac{\sum_{k=1}^K M_k h(J_k)}{\sum_{k=1}^K M_k} = \lim_{K \to \infty} \frac{\sum_{k=1}^K [\frac{h(J_k)}{\alpha(J_k)}]}{\sum_{k=1}^K [\frac{1}{\alpha(J_k)}]} = \pi(h) := \int h(x) \pi(x) \mu(d x)$$
\end{enumerate}

Although we have a solid theory base for Rejection-Free on general state space, applying Rejection-Free to the continuous sampling questions efficiently on most computer hardware is pretty hard. The biggest challenge is the calculation of integration $\int q(x, z) \mu(d z)$. We need many calculations for the numerical integration. In addition, such tasks can hardly be split efficiently into specialized hardware with a reasonable amount of parallel calculating units. At the same time, Unbiased PNS can be surprisingly helpful in this case. As long as the Metropolis algorithm can be applied, PNS can be applied straightforwardly without any calculation of integration. We need to choose the Partial Neighbors Sets $\mathcal{N}_i(x)$ to be a finite subset of all the neighbors $\mathcal{N}(x)$ in Algorithm \ref{alg-unbiasedpns}. We check the performance of our Unbiased PNS on a simple continuous sampling question: the Donuts Example.  

\section{Application to the Donuts Example}
\label{sec-donuts}

Inspired by \cite{chi2022donut}, we use a donuts example to show Unbiased PNS's performance on continuous state space. Suppose we have two independent random variables $\mu$ and $\theta$ where
\begin{equation}
    \mu \sim \text{Normal}^+(\mu_0, \sigma^2), \text{   } \theta \sim \text{Uniform}[0, \pi).
\end{equation}
Here, $\text{Normal}^+$ means the Truncated Normal distribution without the negative tail, and $\pi$ in the Uniform distribution means the circular constant instead of the target density. Then we define two random variables $X_1$ and $X_2$ to be 
\begin{equation}
    X_1 = \sqrt{\mu} \sin \theta, \text{   } X_2 = \sqrt{\mu} \cos \theta.
\end{equation}
The determinant of the Jacobian matrix is $\frac{1}{2}$. Thus we have
\begin{equation}
    f_{X_1, X_2}(x_1, x_2) \propto \frac{1}{\sigma}\exp{\Bigg[-\frac{(x_1^2 + x_2^2 - \mu_0)^2}{2\sigma^2}\Bigg]},
\end{equation}

\begin{figure}
    \centering
    \includegraphics[width= \linewidth]{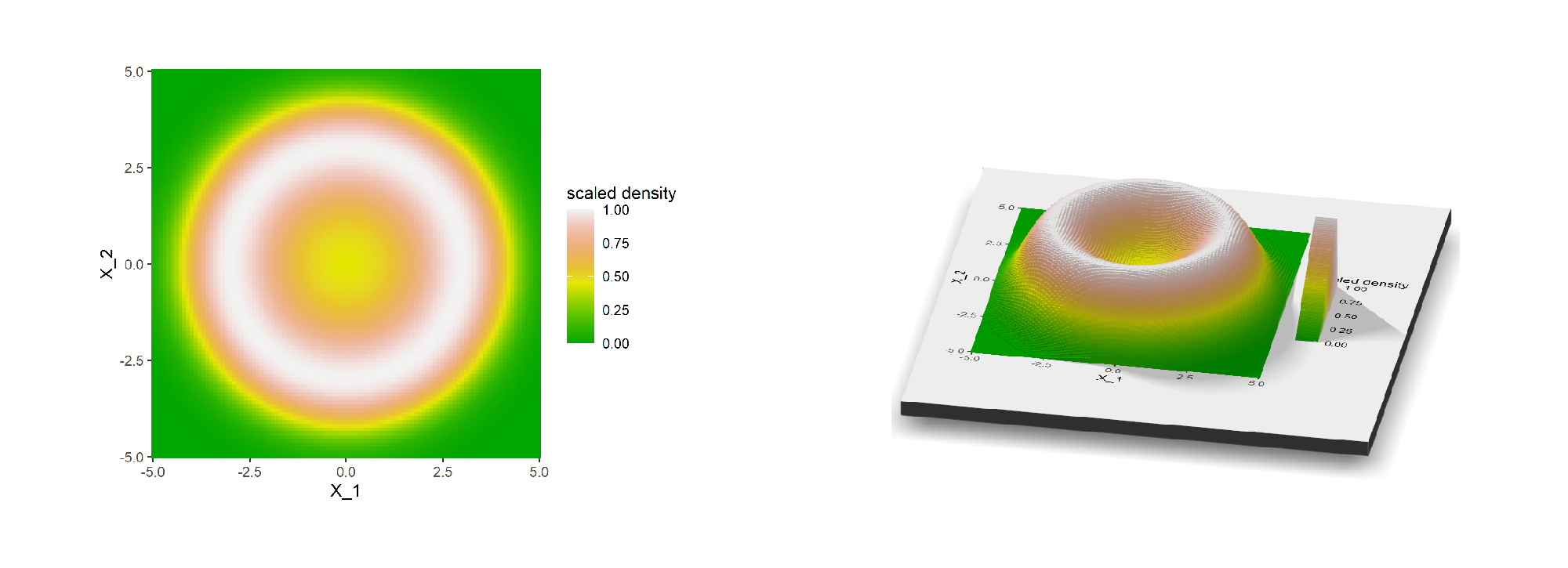}
    \caption{The scaled probability density plot for the Donuts Example with $\mu_0 = 9$ and $\sigma = 10$. The density is scaled to $[0, 1]$. We used large $\sigma$ to show the shape of our distribution. With small $\sigma$, it is hard to see the shape of a sharply peaked distribution.}
    \label{fig-donutsdensity}
\end{figure}

For example, a 3-D map for the density for $X_1$ and $X_2$ with $\mu_0= 9$, and $\sigma = 10$ is shown in Figure \ref{fig-donutsdensity}.  The density is scaled to $[0, 1]$. In our later simulation, we use $\mu_0= 9$ and $\sigma = 0.1$ instead. We use large $\sigma$ to show the shape of our distribution because it is hard to see its shape when it sharply peaks with a small $\sigma$. However, for the simulation, PNS can outperform the Metropolis algorithm when there are many rejections, so we use a small $\sigma$ to get a sharply peaked distribution to increase the rejection rate in the Metropolis algorithm. Note that Unbiased PNS and Rejection-Free are not always better than the Metropolis algorithm. For an extreme example, when we have a distribution where all the states have the same target density values, there will be no rejection for the Metropolis algorithm. At each step, the Metropolis algorithm will uniformly pick a random neighbor from the current state and move to that neighbor, while (Rejection-Free / PNS) will calculate the transition probabilities for (all / part) of the neighbors and uniformly pick a random one. The Metropolis algorithm will be far better than Rejection-Free and PNS in this case. In practice, the higher the dimension of the problem and the more sharply peaked the distribution is, the better the Rejection-Free and Unbiased PNS will be. Thus, we use $\sigma = 0.1$ for the simulation to create a sharply peaked distribution for later simulation. 

In addition, the proposal distribution is defined to be the standard normal distribution for both dimensions. That is, for $x = (x_1, x_2), y = (y_1, y_2) \in \mathbb{R}^2$, $\mathcal{Q}(x, y) = \phi(y_1 - x_1) \phi(y_2 - x_2)$ where $\phi$ is the density function of the standard normal distribution. Then for any $x \in \mathbb{R}^2$, we have $\mathcal{N}(x) = \mathbb{R}^2$.

Since the Partial Neighbor Sets are always the whole space of $\mathbb{R}^2$, it is tough for us to apply Rejection-Free here since the integration of the whole space needs too many computational resources. Even if we limit the neighbors to a small area around the current state, integration is needed as long as the problem is continuous, and the Rejection-Free will be consequentially slow. At the same time, Unbiased PNS can be applied to continuous cases without calculating integration by making minor changes to Algorithm \ref{alg-unbiasedpns}. The Unbiased PNS algorithm for continuous is stated as Algorithm \ref{alg-unbiasedpns-continuous}. In Algorithm \ref{alg-unbiasedpns-continuous}, we did not define the systematic Partial Neighbor Sets as we had for the discrete cases. We want to use Unbiased PNS with finite many partial neighbors being considered at each step, but we have uncountable neighbors. It is impossible to divide these uncountable neighbors into finite partial neighbor sets with finite sizes. Thus, we can only use the random Partial Neighbor Set, which randomizes a new finite partial neighbor set for every $L_0$ original samples. In later simulation, we use partial neighbor sets with $\lvert \mathcal{N}_k \rvert = 50$. That is, we consider 50 partial neighbors at each step. Note that, in Section \ref{subsec-alternating}, we defined the Partial Neighbor Sets, and according to the third condition, we must have reversibility for all $\mathcal{N}_i(x)$, which means $y \in \mathcal{N}_i(x)  \iff x \in \mathcal{N}_i(y)$, $\forall x, y \in \mathcal{S}$. Therefore, we choose the Partial neighbor Set $\mathcal{N}_i(x)$ as follows:
\begin{enumerate}
    \item generate $\delta_1, \delta_2 \sim \mbox{Normal}(0, 1)$;
    \item for state $x = (x_1, x_2)$, put $y = (x_1 + \delta_1, x_2 + \delta_2)$ into the Partial Neighbor Set $\mathcal{N}_i(x)$;
    \item to ensure the reversibility, also put $y' = (x_1 - \delta_1, x_2 - \delta_2)$ into the Partial Neighbor Set $\mathcal{N}_i(x)$;
    \item repeats the above steps 25 times to generate a total of 50 neighbors for the Partial Neighbor Set $\mathcal{N}_i(x)$.
\end{enumerate}
In addition, $L_0$ is selected to be $1000$. Using $L_0 = 100$ to $1000$ will not affect the sampling speed too much, similar to the conclusion in Section \ref{sec-chooseparameters}. 

\begin{algorithm}
\caption{Unbiased PNS for Continuous Case}\label{alg-unbiasedpns-continuous}
\begin{algorithmic}
\State select one Partial Neighbor Set $\mathcal{N}_0$
\State initialize $L \gets L_0$  \Comment{start with $L_0$ remaining original samples}
\State initialize $J_0$
\For{$k$ in $1$ to $K$}
    \State calculate multiplicity list $m \gets 1 + G$ where $G \sim \text{Geometric}(p)$ with 
    $$p =  \sum_{z \in \mathcal{N}_0(J_{k-1})}\mathcal{Q}_0(J_{k-1}, z)\min\Bigg\{1, \frac{\pi(z)\mathcal{Q}_0(z, J_{k-1})}{\pi(J_{k-1})\mathcal{Q}_0(J_{k-1}, z)}\Bigg\}$$
    \If{$m \le L$} \Comment{if we have enough remaining original samples}
        \State  $M_{k-1}  \gets m$, $L \gets L - m$
        \State choose the next jump chain State $J_{k} \in \mathcal{N}_0(J_{k-1})$ such that $$\hat{P}(J_{k} = y \mid J_{k-1}) \propto \mathcal{Q}_0(J_{k-1}, y)\min\Bigg\{1, \frac{\pi(y) \mathcal{Q}_0(y, J_{k-1})}{\pi(J_{k-1}) \mathcal{Q}_0(y, J_{k-1})}\Bigg\}$$
    \Else \Comment{if we don't have enough remaining original samples}
        \State $M_{k-1}  \gets L$, $L \gets L_0$, $J_{k} \gets J_{k-1}$, \State \Comment{stay at $J_{k-1}$ for the remaining $L$ times}
        \State select a new Partial Neighbor Set $\mathcal{N}_0$
    \EndIf
\EndFor
\end{algorithmic}
\end{algorithm}

Moreover, we measure the sampling results by bias instead of the TVD. The calculation of TVD in the continuous case also needs much integration, which is hard to calculate. On the other hand, given samples $\{X_1, X_2. \dots, X_K\}$, we usually use the MCMC to approximate the expected value $\mathbb{E}_\pi(h)$ of a function $h : S \to \mathbb{R}$ by the usual estimator, $\hat{e}_K(h) = \frac{1}{K}\sum_{k=1}^K h(X_{1, k}, X_{2, k})$. The Strong Law of Large Numbers for Markov chains says that assuming that $\mathbb{E}_\pi(h)$ is finite and that the Markov chain is irreducible with stationary distribution $\pi$, we must have $\lim_{K \to \infty} \hat{e}_K = \mathbb{E}_\pi(h)$. Therefore, $\text{Bias}(h) = \lvert \hat{e}_K(h) - \mathbb{E}_\pi(h)\lvert  = \lvert \frac{1}{K}\sum_{k=1}^K h(X_k) - \mathbb{E}_\pi(h) \lvert$ can also be a good measurement for the quality of the samples. According to the definition, bias is greater or equal to 0. When the samples $\{X_1, X_2. \dots, X_K\}$ gets closer to the target distribution $\pi$, the bias will decrease to 0. Thus, convergence to stationarity is described by how quickly the bias decreases to 0 for all function $h$. This property is similar to TVD from Section~\ref{sec-qubo}. In fact, for any probability distribution $\mathcal{P}_1$ and $\mathcal{P}_2$, $\text{TVD}(\mathcal{P}_1, \mathcal{P}_2) = \sup_\mathcal{S}\big(\mathcal{P}_1(\mathcal{S}), \mathcal{P}_2(\mathcal{S})\big)$ \citep{chen2016general}. 

For example, we check the sum of the bias from the first-degree terms $X_1$ and $X_2$. Since the Donuts example is centered at $0$, thus $\mathbb{E}_\pi(X_1) = \mathbb{E}_\pi(X_2) = 0$. Thus, we have 
\begin{equation}
\begin{aligned}
    \text{Bias}(X_1) + \text{Bias}(X_2) 
    & =  \hat{e}_K(X_1) - \mathbb{E}_\pi(X_1) + \hat{e}_K(X_2) - \mathbb{E}_\pi(X_2) \\
    & =  \bigg\lvert \frac{1}{K}\sum_{k=1}^K X_{1, k}   \bigg\rvert +  \bigg\lvert \frac{1}{K}\sum_{k=1}^K X_{2, k} \bigg\rvert,
\end{aligned}
\end{equation}
Note that both biases will decrease to $0$ if our Markov chain converges to the target density $\pi$. In addition, for the Rejection-Free Chain $\{J_{1,k}, J_{2, k}, M_k\}_{k=1}^K$ generated by the Unbiased PNS algorithm, the bias is defined to be  
\begin{equation}
    \text{Bias}(J_1) + \text{Bias}(J_2) = \frac{\big \lvert \sum_{k=1}^K M_k \times J_{1, k} \big \rvert}{\sum_{k=1}^K M_k} + \frac{\lvert \sum_{k=1}^K M_k \times J_{2, k} \big \rvert}{\sum_{k=1}^K M_k}
\end{equation}
The result for comparing the Metropolis algorithm and Unbiased PNS by the bias of first-degree term is shown in Figure~\ref{fig-donutresult001}. Each dot within the plot represents the average value of 100 simulation runs. For each run, we generate a Markov chain for a given number of samples for both algorithms. The average time represents the CPU time we apply the algorithm by a single-core implementation. Again, parallelism hardware such as DA can yield 100x to 10,000x speedups for Unbiased PNS \citep{sheikholeslami2021power}.

\begin{figure}
    \centering
    \includegraphics[width= \linewidth]{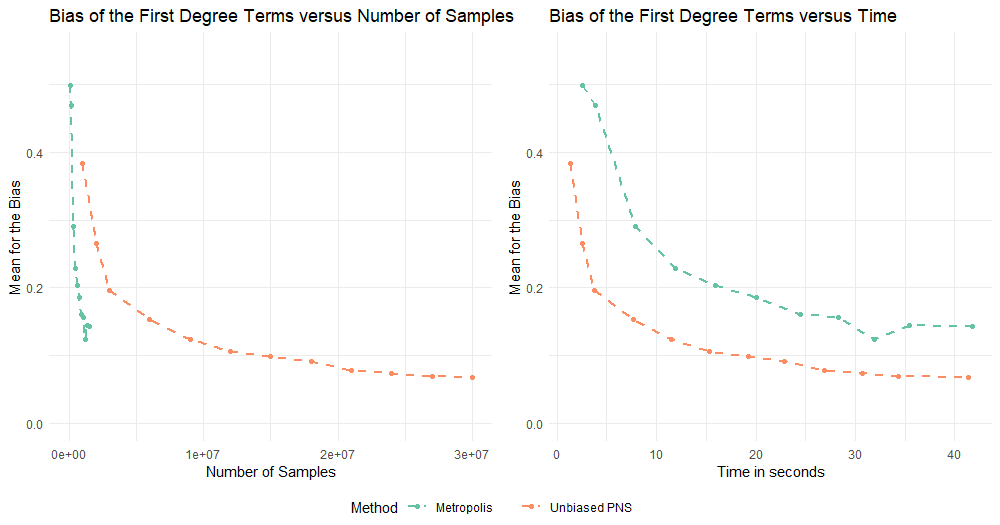}
    \caption{Sum of the Average Bias of $X_1$ and $X_2$ between sampling and target density $\pi$ as a function of the number of iterations (left) and average time in seconds (right) for two methods: Metropolis algorithm and Unbiased PNS. We used Donuts example with $r_0 = 10$ and $\sigma = 0.1$. Each dot within the plot represents the result of the average bias value and time used for 1000 simulation runs given certain original sample sizes. The original sample sizes for the Metropolis algorithm are $\{50000, 100000, 150000, 300000, 450000, \dots, 1500000\}$. The sizes for Unbiased PNS are 20x more than the sizes for the Metropolis. We choose these sizes to get a close average CPU time for both methods. For Unbiased PNS, we used $\lvert \mathcal{N}_k \rvert = 50$ and $L_0 = 1000$.}
    \label{fig-donutresult001}
\end{figure}   

From Figure~\ref{fig-donutresult001}, we can see that the quality of the samples by Unbiased PNS is again worse than the Metropolis algorithm because each Partial Neighbors Set is biased within $L_0$ original samples, while the combination of them is unbiased. Thus, the average bias values for Unbiased PNS are more significant for the same amount of samples. However, for a given amount of CPU time, the performance of Unbiased PNS is much better than the Metropolis algorithm. For this example, the Unbiased PNS can get 30x more samples than the Metropolis algorithm within the same time by a single-core implementation. Rejections slow down the Metropolis algorithm while Unbiased PNS is not influenced, and thus, Unbiased PNS works much better in this simulation.

In addition, we can also check the sum of the bias from the second degree terms $\text{Bias}(X_1^2) + \text{Bias}(X_2^2)$, the sum of the bias from the fourth degree terms $\text{Bias}(X_1^4) + \text{Bias}(X_2^4)$, and the sum of the bias from the positive rate $\text{Bias}(\mathbbm{1}(X_1 > 0)) + \text{Bias}(\mathbbm{I}(X_1 > 0))$, where $\mathbbm{1}$ means the indicator function. To calculate the bias of the second degree terms , we have $X_1^2 + X_2^2 = \mu \sim \text{Normal}^+(\mu_0, \sigma^2)$. Note that, for the Truncated normal distribution with mean $9$ and standard deviation $0.1$, the probability for a negative tail is too small, so we can treat it as a normal distribution. Thus,
\begin{equation}
    \mathbb{E}_\pi(X_1^2) = \frac{1}{2}\mathbb{E}_\pi(X_1^2+X_2^2) = \frac{1}{2}\mathbb{E}_\pi(\mu^2) \approx \frac{1}{2} \mu_0^2.
\end{equation}
\begin{equation}
\begin{aligned}
    \text{Bias}(X_1^2) + \text{Bias}(X_2^2) 
    & =  \hat{e}_K(X_1^2) - \mathbb{E}_\pi(X_1^2) + \hat{e}_K(X_2^2) - \mathbb{E}_\pi(X_2^2) \\
    & \approx  \lvert \frac{1}{K}\sum_{k=1}^K X_{1, k}^2 - \frac{1}{2} \mu_0^2  \rvert + \lvert \frac{1}{K}\sum_{k=1}^K X_{2, k}^2 - \frac{1}{2} \mu_0^2 \rvert.
\end{aligned}
\end{equation}
Similarly, 
\begin{equation}
\begin{aligned}
    \text{Bias}(X_1^4) + \text{Bias}(X_2^4) 
    =  & \hat{e}_K(X_1^4) - \mathbb{E}_\pi(X_1^4) + \hat{e}_K(X_2^4) - \mathbb{E}_\pi(X_2^4) \\
    \approx  & \lvert \frac{1}{K}\sum_{k=1}^K X_{1, k}^4 - \frac{3}{8}(\mu_0^4 + \sigma^2)  \rvert + \\ & \lvert \frac{1}{K}\sum_{k=1}^K X_{2, k}^4 - \frac{3}{8}(\mu_0^4 + \sigma^2) \rvert;
\end{aligned}
\end{equation}
\begin{equation}
\begin{aligned}
    \text{Bias}(\mathbbm{1}(X_1 > 0)) + \text{Bias}(\mathbbm{1}(X_2 > 0)) 
    =  & \hat{e}_K(\mathbbm{1}(X_1 > 0)) - \mathbb{E}_\pi(\mathbbm{1}(X_1 > 0)) + \\    
       &  \hat{e}_K(\mathbbm{1}(X_2 > 0)) - \mathbb{E}_\pi(\mathbbm{1}(X_2 > 0)) \\
    = & \lvert \frac{1}{K}\sum_{k=1}^K \mathbbm{1}(X_{1, k} > 0) - \frac{1}{2}  \rvert 
    + \\ & \lvert \frac{1}{K}\sum_{k=1}^K \mathbbm{1}(X_{2, k} > 0) - \frac{1}{2}  \rvert.
\end{aligned}
\end{equation}
\begin{figure}
    \centering
    \includegraphics[width= \linewidth]{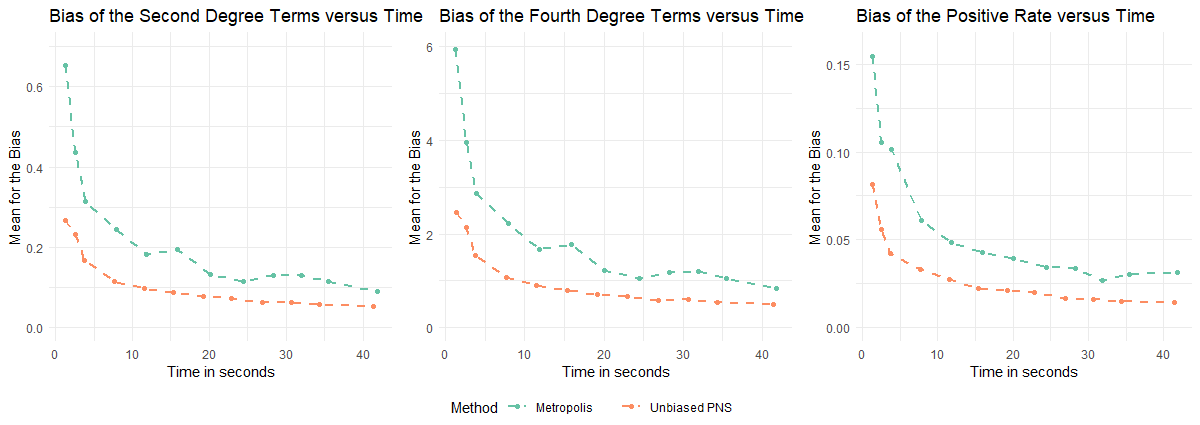}
    \caption{Sum of the average bias from the second degree terms $\text{Bias}(X_1^2) + \text{Bias}(X_2^2)$ (left), the fourth degree terms $\text{Bias}(X_1^4) + \text{Bias}(X_2^4)$ (middle), and the positive rate $\text{Bias}(\mathbbm{1}(X_1 > 0)) + \text{Bias}(\mathbbm{1}(X_1 > 0))$ (right) between sampling and target density $\pi$ as a function of average time in seconds for two methods: Metropolis algorithm and Unbiased PNS. $\mathbb{I}$ means the indicator function. We used Donuts example with $r_0 = 10$ and $\sigma = 0.1$. Each dot within the plot represents the result of the average bias value and time used for 1000 simulation runs given certain original sample sizes. The original sample sizes for the Metropolis algorithm are $\{50000, 100000, 150000, 300000, 450000, \dots, 1500000\}$. The sizes for Unbiased PNS are 20x more than the sizes for the Metropolis. We choose these sizes to get a close average CPU time for both methods. For Unbiased PNS, we used $\lvert \mathcal{N}_k \rvert = 50$ and $L_0 = 1000$.}
    \label{fig-donutresult002}
\end{figure}   

The results for the comparison of the Metropolis algorithm and Unbiased PNS by the sum of the average bias from the second degree terms $\text{Bias}(X_1^2) + \text{Bias}(X_2^2)$, the fourth degree terms $\text{Bias}(X_1^4) + \text{Bias}(X_2^4)$, and the positive rate $\text{Bias}(\mathbb{I}(X_1 > 0)) + \text{Bias}(\mathbb{I}(X_1 > 0))$ are shown in Figure \ref{fig-donutresult002}. From the result for different choices of the terms, we can conclude that Unbiased PNS performs better than the Metropolis algorithm in this continuous Donuts example. 

\section{Burn In by Partial Neighbor Search}
\label{sec-burninpns}

\subsection{Optimization instead of Burn-In}
\label{sec-optimization}
In previous sections, when we need $K$ original samples, we have to generate $2K$ original samples. We use the first $K$ original samples as the burn-in part. This way, we started our sampling process from stationarity. However, \cite{geyer2011introduction} argued that burn-in is not necessary for MCMC. As an alternative to burn-in, any point the researcher does not mind having in a sample is a good starting point. His argument indicates that we can usually start at a point whose target density value is large. \cite{geyer2011introduction} claimed that this alternative method is usually better than regular burn-in.

We can apply optimization algorithms in \cite{chen2022optimization} before sampling if we accept the above statement. For example, we can cancel the burn-in part of the QUBO question in Section \ref{sec-qubo}. Instead, before we start sampling from the target density, we consider optimization algorithms which tries to maximize $\pi(x) = \exp\{x^T Q x\}$ for $x \in \{0, 1\}^N$. Then we can start sampling from a state with a large $\pi(x)$ value, although it may not be optimal. Simulated Annealing is one such algorithm. In addition, we can also use Optimization Rejection-Free and Optimization PNS from \cite{chen2022optimization}.

To find $x$ from the state space $\mathcal{S}$ which maximizes $\pi(x)$, given the proposal distribution $\mathcal{Q}$, and the corresponding neighbors $\mathcal{N}$, and a non-increasing cooling schedule $T: \mathbb{N} \to (0, \infty)$, the corresponding algorithms for Simulated Annealing, Optimization Rejection-Free, and Optimization PNS are described in Algorithm \ref{alg-sa}, \ref{alg-optrf}, and \ref{alg-optpns}.

\begin{algorithm}
\caption{Simulated Annealing}\label{alg-sa}
\begin{algorithmic}
\State initialize $X_0$, and $X_{\max} = X_0$
\For{$k$ in $1$ to $K$}
    \State random $Y \in \mathcal{N}(X_{k-1})$ based on $\mathcal{Q}(X_{k-1}, \cdot)$
    \State random $U_k \sim \text{Uniform}(0, 1)$
    \If{$U_k < \Large[\frac{\pi(Y)}{\pi(X_{k-1})}\Large]^{{1 / T(k)}}$} 
        \State \Comment{accept with probability $\min \Big{\{} 1, \big{[}\frac{\pi(Y_k)}{\pi(X_{k-1})}\big{]}^{{1 / T(k)}} \Big{\}} $}
        \State $X_{k} = Y$ \Comment{accept and move to state $Y$}
        \If{$\pi(Y) > \pi(X_{\max})$}
            \State $X_{\max} = Y$
        \EndIf
    \Else \State $X_{k} = X_{k-1}$ \Comment{reject and stay at $X_{k-1}$}
    \EndIf 
\EndFor
\end{algorithmic}
\end{algorithm}

\begin{algorithm}
\caption{Optimization Rejection-Free}\label{alg-optrf}
\begin{algorithmic}
\State initialize $J_0$, and set $X_{\max} = J_0$
\For{$k$ in $1$ to $K$}
    \State choose the next jump chain State $J_{k} \in \mathcal{N}(J_{k-1})$ such that $$\hat{P}(J_{k} = y \mid J_{k-1}) \propto \mathcal{Q}(J_{k-1}, y)\min\Bigg\{1, \frac{\pi(y) \mathcal{Q}(y, J_{k-1})}{\pi(J_{k-1}) \mathcal{Q}(J_{k-1}, y)}\Bigg\}$$
    \If{$\pi(J_k) > \pi(X_{\max})$}
        \State $X_{\max} = J_k$
    \EndIf
\EndFor
\end{algorithmic}
\end{algorithm}

\begin{algorithm}
\caption{Optimization Partial Neighbor Search}\label{alg-optpns}
\begin{algorithmic}
\State initialize $J_0$, and set $X_{\max} = J_0$
\For{$k$ in $1$ to $K$}
    \State pick the Partial Neighbor Set $\mathcal{N}_k(J_{k-1}) \subset \mathcal{N}(J_{k-1})$
    \State choose the next jump chain State $J_{k} \in \mathcal{N}_k(J_{k-1})$ such that $$\hat{P}(J_{k} = y \mid J_{k-1}) \propto \mathcal{Q}(J_{k-1}, y)\min\Bigg\{1, \frac{\pi(y) \mathcal{Q}(y, J_{k-1})}{\pi(J_{k-1}) \mathcal{Q}(y, J_{k-1})}\Bigg\}$$
    \If{$\pi(J_k) > \pi(X_{\max})$}
        \State $X_{\max} = J_k$
    \EndIf
\EndFor
\end{algorithmic}
\end{algorithm}

In \cite{chen2022optimization}, we illustrated the superior performance of Optimization PNS with many examples, such as the QUBO question, the Knapsack problem, and the 3R3XOR problem. In all these problems, Optimization PNS is the best algorithm compared to the Simulated Annealing algorithm and Optimization Rejection-Free. See \cite{chen2022optimization} for more details. Therefore, we can use Optimization PNS as in Algorithm \ref{alg-optpns} to replace the burn-in part before sampling. 

However, the starting states obtained by the proposed three optimization algorithms will not converge to the target density. Therefore, people can use these optimization methods to replace the burn-in part only if they believe that the sampling can start without stationarity, just like \cite{geyer2011introduction}.

\subsection{Burn-In until Convergence}
\label{sec-convergence}

Section \ref{sec-optimization}, we mentioned that some people believe that MCMC does not necessarily need to start from stationarity. However, some people may insist on starting from stationarity. Then we can combine the algorithm for optimization and sampling and try to take advantage of both versions to get a burn-in algorithm. We can apply the optimization algorithm for a certain number of steps $K_0$, and then we apply the sampling algorithms such as Rejection-Free (Algorithm \ref{alg-rf}) or Unbiased PNS (Algorithm \ref{alg-unbiasedpns}) for $K_1$ samples.

To check the distribution of the states after a certain number of steps of the hybrid algorithm, We generate a certain number of Markov chains by the algorithms and record each chain's last state. As a result, we can get the distribution after burn-in, and we call this distribution the starting distribution for sampling. For example, just like the previous example in Section \ref{sec-qubo}, we still consider a $16\times16$ QUBO question. Every setting is exactly the same as what we have in Section~\ref{sec-qubo} except we used $Q_{i, j} \sim \text{Normal}(0, 1^2)\mbox{, } \forall i \le j$. We didn't use the standard deviation of $10$ like Section \ref{sec-qubo}, since we only use the last states from one Markov chain, we generate one such starting distribution with $100,000$ Markov chains and check the TVD value between the starting distribution and the target density. Thus, if we use a standard deviation of $10$, we need much a much longer time to get a small TVD value. The number of steps for Optimization PNS $K_0$ is chosen to be $\lfloor \frac{1}{20}K_1 \rfloor$ ($\lfloor$  $\rfloor$ represents the floor function). The number of samples $K_1 = 20, 40, 60, \dots, 200$. In addition, we also compare Unbiased PNS with Optimization PNS plus Unbiased PNS. Since we believe Unbiased PNS will converge slower than Rejection-Free, so we choose $K_1 = 40, 80, \dots, 600$, and $K_0 = \lfloor \frac{1}{40}K_1 \rfloor$. The temperature function $T(k)$ in the optimization algorithms is set to be constantly $1$. The result is shown in Figure \ref{fig-convergence}.

\begin{figure}
    \centering
    \includegraphics[width= \linewidth]{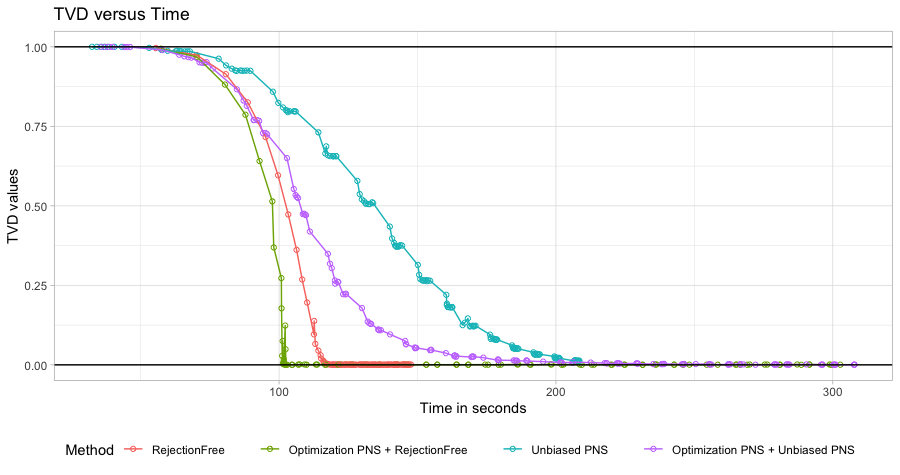}
    \caption{Average values of TVD between the starting distribution from $100,000$ chains and target density $\pi$ as a function of the average time for the chains in seconds for four methods: Rejection-Free, Optimization PNS plus Rejection-Free, Unbiased PNS, and  Optimization PNS plus Unbiased PNS. Random upper triangular $16 \times 16$ QUBO matrix is generated randomly by $Q_{i, j} \sim N(0, 1^2)$ for upper triangular elements. The original sample sizes for Rejection Free are $K_1 = \{20, 30, 40, 50, \dots, 1000\}$, and the number of steps for the corresponding Optimization PNS is $K_0 =  \lfloor \frac{K_1}{20}\rfloor$. The original sample sizes for Unbiased PNS are $K_1 = \{40, 50, 60 , \dots, 1500\}$, and the number of steps the corresponding Optimization PNS is $K_0 = \lfloor \frac{K_1}{40}\rfloor$. Each dot within the plot represents the TVD value between the target distribution $\pi$ and the distribution of the last state of $100,000$ Markov chains.}
    \label{fig-convergence}
\end{figure}

From Figure \ref{fig-convergence}, algorithms with the help of Optimization PNS converge faster with respect to the CPU time. We used a $16 \times 16$ QUBO question here. In addition, we concluded that the higher the dimension is and the more sharply peaked the distribution is, the better the optimization PNS will be \citep{chen2022optimization}. Optimization PNS performs extremely well in the optimization version of $200 \times 200$ QUBO question \citep{chen2022optimization}. Thus, we can also use the Optimization PNS to help burn-in in high dimension or sharply peaked distributions. Since the calculation of TVD for high dimension problems is infeasible, so we just did a simulation of $16 \times 16$ QUBO question here. See \cite{chen2022optimization} for more simulation results from optimization problems with higher dimensions. 

\section{Conclusion}
\label{sec-conclusion}

We introduced three versions of the Partial Neighbor Search algorithms of sampling. Basic PNS is straightforward but does not converge to the target density. The Unbiased PNS will converge to the target density, but it performs worse than Rejection-Free compared to a single-core implementation in the QUBO question. However, the Unbiased PNS can use specialized parallelism hardware such as DA to improve the sampling efficiency significantly, while Rejection-Free cannot. In addition, Rejection-Free is infeasible in many continuous cases, but the Unbiased PNS can be applied to all continuous cases and works much better than the Metropolis algorithm. Finally, we illustrated that Optimization PNS from \cite{chen2022optimization} can be used to improve the burn-in part before sampling. 
   
\section*{Acknowledgments}

The authors thank Fujitsu Ltd. and Fujitsu Consulting (Canada) Inc. for providing financial support.

\bigskip

\bibliography{References}

\appendix
\section{Unbiased PNS Convergence Theorem}
\label{appendixA}

\begin{definition}
\label{def-1}
For sampling questions in general state space, we usually have the following elements:
\begin{enumerate}
    \item a state space $\mathcal{S}$;
    \item a $\sigma$-finite reference measure $\mu$ on $\mathcal{S}$, where $\mu$ could be a counting measure for discrete cases, and $\mu$ could be a Lebesgue measure for continuous cases;
    \item a target density  $\pi: \mathcal{S} \to [0, 1]$, where $\int_{x \in \mathcal{S}} \pi(x) \mu(d x) = 1$;
    \item a target distribution $\Pi: \mathbf{P}(\mathcal{S}) \to [0, 1]$, where $\Pi(\mathcal{A}) \coloneqq \int_\mathcal{A} \pi(x) \mu(d x)$, $\forall \mathcal{A} \subset \mathcal{S}$, and $\mathbf{P}$ means the power set;
    \item a proposal density $q(x, y): \mathcal{S} \times \mathcal{S} \to [0, 1]$, where $\int_{\mathcal{S}} q(x, y) \mu(d y) = 1$, $\forall x, y \in \mathcal{S}$;
    \item a proposal distribution $\mathcal{Q}(x, d y) \propto q(x, y) \mu(d y)$;
    \item a corresponding neighbor set $\mathcal{N}(x) \coloneqq \{y \in \mathcal{S} \mid q(x, y) > 0\} \subset \mathcal{S} \backslash \{x\}$;
    \item the transition probabilities $P(x, d y) = q(x, y) \min \Big(1, \frac{\pi(y) q(y, x)}{\pi(x) q(x, y)}\Big) \mu(d y)$, where $P(x, d y) = q(x, y)\mu(d y)$ if the denominator $\pi(x) q(x, y) = 0$.
\end{enumerate}
Given the above elements, assume irreducibility and aperiodicity, we can generate a Markov chain $\{X_0, X_1, \dots, X_K\}$ such that the limiting distribution of $\lim_{n \to \infty} X_n$ converges to the stationarity distribution $\pi(x)\mu(d x)$ by Algorithm \ref{alg-metropolis}.
\end{definition}

\begin{definition}
\label{def-2}
Suppose we have a state space $\mathcal{S}$, a reference measure $\mu$, and a target density $\pi$, the proposal distribution $\mathcal{Q}$ and the corresponding neighbor set $\mathcal{N}$. Then, a Partial Neighbor Set $\mathcal{N}_i$ means a function $\mathcal{N}_i$ satisfying the following conditions:
\begin{enumerate}
    \item $\mathcal{N}_i: \mathcal{S} \to \mathbf{P}(\mathcal{S})$, where $\mathcal{S}$ is the state space, and $\mathbf{P}(\mathcal{S})$ is the power set of $\mathcal{S}$;
    \item $\mathcal{N}_i(x) \subset \mathcal{N}(x)$, $\forall x \in \mathcal{S}$, and we must pick a finite subset $\mathcal{N}_i(x)$ to ensure a finite for loop in Algorithm \ref{alg-unbiasedpns};
    \item $y \in \mathcal{N}_i(x)  \iff x \in \mathcal{N}_i(y)$, $\forall x, y \in \mathcal{S}$;
\end{enumerate}
\end{definition}
Given a Partial Neighbor Set $\mathcal{N}_i$, the proposal distribution for $\mathcal{N}_i$ is defined to be $\mathcal{Q}_i: \mathcal{S} \times \mathbf{P}(\mathcal{S}) \to \mathbb{R}$, where $\mathcal{Q}_i(x, d y) =\frac{\sum_{r \in \mathcal{N}_i}q(x, r)\delta_r(d y)}{\sum_{z \in \mathcal{N}_i}q(x, z)}$, where $\delta_r$ means the point mass at $r$.

Here, before we prove the convergence theorem of the Unbiased PNS as stated in Algorithm \ref{alg-unbiasedpns}, we first prove it for another version of the Unbiased PNS as stated in Algorithm \ref{alg-unbiasedpns-naive}. It is easy to see that the only difference between Algorithm \ref{alg-unbiasedpns-naive} and Algorithm \ref{alg-unbiasedpns} is that we are not using the Rejection-Free technique here, where we calculate all the transition probabilities at once, pick the next jump chain state, and calculate the multiplicity list according to the transition probabilities.

\begin{algorithm}
\caption{Unbiased Partial Neighbor Search without Rejection-Free technique}\label{alg-unbiasedpns-naive}
\begin{algorithmic}
\State select $\mathcal{N}_i$ for $i = 0, 1, \dots, \mathcal{I}-1$ where $\cup_{i=0}^{\mathcal{I}-1} \mathcal{N}_i(X)= \mathcal{N}(X)$
\State initialize $i = 0$ \Comment{start with neighbor set $\mathcal{N}_0$}
\State initialize $L = L_0$ \Comment{start with $L_0$ remaining samples}
\State initialize $X_0$\Comment{initial the starting state}
\For{$k$ in $1$ to $K$}
\State random $Y \in \mathcal{N}_i(J_{k-1})$ based on $\mathcal{Q}_i(X_{k-1}, \cdot)$
\State random $U_k \sim \text{Uniform}(0, 1)$
    \If{$U_k < \frac{\pi(Y) \mathcal{Q}_i(Y, X_{k-1})}{\pi(X_{k-1}) \mathcal{Q}_i(X_{k-1}, Y)}$}     
    \State \Comment{accept with probability $\min \Big{\{} 1,  \frac{\pi(Y) \mathcal{Q}_i(Y, X_{k-1})}{\pi(X_{k-1}) \mathcal{Q}_i(X_{k-1}, Y)} \Big{\}} $}
    \State $X_{k} = Y$ \Comment{accept and move to state $Y$}
    \Else \State $X_{k} = X_{k-1}$ \Comment{reject and stay at $X_{k-1}$}
    \EndIf
    \State L = L - 1 \Comment{a new sample from $\mathcal{N}_i$}
    \If{$L = 0$} \Comment{if we don't have enough remaining samples}
        \State $L = L_0$, and $i = {i + 1 \mod \mathcal{I}}$ \Comment{switch to the next $\mathcal{N}_i$}
    \EndIf
\EndFor
\end{algorithmic}
\end{algorithm}

\begin{proposition}
\label{prop-1}
Suppose we have a state space $\mathcal{S}$, a reference measure $\mu$, and a target density $\pi$, the proposal distribution $\mathcal{Q}$ and the corresponding neighbor set $\mathcal{N}$. In addition, suppose the Partial Neighbor Set $\{\mathcal{N}_i\}_{i=0}^{\mathcal{I}-1}$ satisfies all the conditions in Definition \ref{def-1}. Then $\pi(x)\mu(d x)$ is the stationary distribution for Algorithm \ref{alg-unbiasedpns-naive} with the partial neighbor set $\mathcal{N}_i$.
\end{proposition}
\begin{proof}
Let $P_i(x, d y)$ be the transition probability for Partial Neighbor Set $\mathcal{N}_i$. Then $\forall y \in \mathcal{N}_i(x)$ where $q(x, y) > 0$, we have

\begin{equation}
\begin{aligned}
    \pi(x) \mu(d x) P_i(x, d y) 
    & = \pi(x) \mu(d x) q(x, d y) \min \bigg(1, \frac{\pi(y) q(y, x) \mu(d x)}{\pi(x) q(x, d y) \mu(d y)}\bigg) \\ 
    & = \min \Big(\pi(x) \mu(d x) q(x, d y), \pi(y) q(y, x) \mu(d x)\Big) \\
    & = \pi(y) \mu(d y) P_i(y, d x) \\
\end{aligned}
\end{equation}
Thus, by reversibility, $\mathcal{N}_i$ is stationary with $\pi(x)\mu(d x)$.
\end{proof}

\begin{proposition}
\label{prop-2}
Suppose we have a state space $\mathcal{S}$, a reference measure $\mu$, a target density $\pi$, and a Markov chain $\{X_0, X_1, X_2, \dots\}$ produced by algorithm \ref{alg-unbiasedpns-naive}. In addition, suppose $\pi(x)\mu(d x)$ is the stationary distribution  is the stationary distribution for Algorithm \ref{alg-unbiasedpns-naive} with all $\{\mathcal{N}_i\}_{i=0}^{\mathcal{I}-1}$, and $\union_{i=0}^{\mathcal{I}-1} \mathcal{N}_i$ makes the Markov chain irreducible. Moreover, suppose there are rejections for the Markov chain, and thus the Markov chain is aperiodic. Then the Markov chain converges in total variation distance; i.e.:
\begin{equation}
    \lim_{k \to \infty} \sup_{\mathcal{A} \subset \mathcal{S}} \Big\lvert P(X_k \in \mathcal{A}) - \int_\mathcal{A} \pi(y) \mu(d y) \Big \rvert = 0
\end{equation}

\end{proposition}
\begin{proof}
This follows immediately from Theorem 13.0.1 in \cite{meyn2012markov}.
\end{proof}

\begin{theorem}
\label{thm-3}
Suppose we have a state space $\mathcal{S}$, a reference measure $\mu$, a target density $\pi$, a Markov chain $\{X_0, X_1, X_2, \dots\}$ produced by algorithm \ref{alg-unbiasedpns-naive}, and a jump chain $\{(J_0, M_0), (J_1, M_1), (J_2, M_2), \dots\}$ produced by algorithm \ref{alg-unbiasedpns}. Meanwhile, suppose the proposal distribution $\mathcal{Q}$ and the corresponding neighbor set $\mathcal{N}$ ensure the Markov chain produced by the Metropolis-Hastings algorithm converges to the stationarity $\pi(x)\mu(d x)$. In addition, suppose $\pi(x)\mu(d x)$ is the stationary distribution for all $\{\mathcal{N}_i\}_{i=0}^{\mathcal{I}-1}$, and $\union_{i=0}^{\mathcal{I}-1} \mathcal{N}_i$ makes both chains irreducible. Moreover, suppose both chains are aperiodic. Then the jump chain has the following properties:
\begin{enumerate}
    \item the transition probability $P_i$ from the Markov chain and the transition probability $\hat{P}_i$ from the jump chain satisfy $\hat{P}_i(x, d y) = \frac{1}{\alpha(x)} P(x, d y) \mathbbm{1}(x \ne y)$, and $\hat{P}_i(x, \{x\}) = 0$;
    \item  The conditional distribution of $M_k$ given $J_k$ is equal to the distribution of $1+G$ where $G$ is a geometric random variable with success probability $\alpha(J_k)$ where $\alpha(x) \coloneqq 1 - P_i(x, \{x\})$;
    \item If the original chain is $\phi$-irreducible (see, e.g., \cite{meyn2012markov}) for some positive $\sigma$-finite measure $\phi$ on $\mathcal{S}$, then the jump chain is also $\phi$-irreducible for the same $\phi$.
    \item  If the Markov chain has stationary distribution $\pi(x) \mu(d x)$, then the jump chain has stationary distribution given by $\hat{\pi}(x) = c \alpha(x) \pi(x) \mu(d x)$ where $c^{-1} = \int \alpha(y) \pi(y) \mu(d y)$
    \item If $h : \mathcal{S} \to \mathbb{R}$ has finite expectation, then with probability $1$,
    $$\lim_{K \to \infty} \frac{\sum_{k=1}^K M_k h(J_k)}{\sum_{k=1}^K M_k} = \lim_{K \to \infty} \frac{\sum_{k=1}^K [\frac{h(J_k)}{\alpha(J_k)}]}{\sum_{k=1}^K [\frac{1}{\alpha(J_k)}]} = \pi(h) := \int h(x) \pi(x) \mu(d x)$$
\end{enumerate}
\end{theorem}
\begin{proof}
The proof is trivial given the Proposition \ref{prop-2} and Theorem 13 from the Rejection-Free paper \citep{rosenthal2021jump}. We reviewed Theorem 13 from the Rejection-Free paper in Section \ref{sec-continuous}.
\end{proof}

\section{How to Sample Proportionally}
\label{appendixC}

Given $A_i > 0$, for $i = 1, 2, \dots, N$, how can we sample $Z$ so that $P(Z = i) = \frac{A_i}{\sum_j A_j}$? We could choose $U \sim \mbox{Uniform}[0, 1]$, and then set $Z = \min\{i, \sum_{j=1}^i A_j > U \times \sum_{j=1}^N A_j\}$. However, this involves summing all of the $A_j$ , which is inefficient. If $\sum_{j=1}^N A_j = 1$, then we could choose $U \sim \mbox{Uniform}[0, 1]$ and just set $Z = \min\{i, \sum_{j=1}^i A_j > U\}$, which is slightly easier, and can be done by binary searching. However, it still requires summing lots of the $A_j$, which could still be inefficient. If $\sum_{j=1}^N A_j < 1$, then we could choose $U \sim \mbox{Uniform}[0, 1]$, and then still set$Z = \max\{i, \sum_{j=1}^i A_j > U\}$, except if no such $i$ exists then we reject that choice of $U$ and start again. In addition to the previous problems, this could involve lots of rejection if $\sum_{j=1}^N A_j$ is much smaller than 1, which is again inefficient. Another option is the following method, based on \cite{efraimidis2006weighted}; see also the n-fold way approach to kinetic Monte Carlo in \cite{bortz1975new}.

\begin{proposition}
\label{prop-5}
Let $A_1, A_2, \dots, A_N$ be positive numbers, Let $\{R_j\}_{j=1}^N$ be i.i.d. $\sim \mbox{Uniform}[0, 1]$, and let $d_j = - \frac{\log(R_j)}{A_j}$ for $j = 1,2,\dots,N$. Finally, set $Z = \arg \min_j d_j$. Then $P[Z = i] = \frac{A_i}{\sum_j A_j}$, i.e. $Z$ selects $i$ from $\{1,2,\dots,N\}$ with probability proportional to $A_i$. 
\end{proposition}

\begin{proof}
\begin{equation}
\begin{aligned}
    P[Z=i] 
    & = P[d_j > d_i, \forall j \ne i] \\
    & = P[- \frac{\log(R_j)}{A_j} > - \frac{\log(R_i)}{A_i}, \forall j \ne i] \\
    & = P[R_j < R_i^{A_j / A_i}, \forall j \ne i] \\
    & = \int_{0}^1 P[R_j < R_i^{A_j / A_i}, \forall j \ne i \mid R_i = x] d x \\
    & = \int_{0}^1 P[R_j < x^{A_j / A_i}, \forall j \ne i] d x \\
    & = \int_{0}^1\prod_{j \ne i} x^{A_j / A_i} d x \\
    & = \int_{0}^1 x^{\sum_{j \ne i} A_j / A_i} d x \\
    & = \frac{x^{[\sum_{j \ne i} A_j / A_i + 1]}}{\sum_{j \ne i} A_j / A_i + 1} \Bigg{\rvert}_{x=0}^1 \\
    & = \frac{A_i}{\sum_j A_j}
\end{aligned}
\end{equation}
\end{proof}

Proposition \ref{prop-5} is useful, especially when we apply Rejection-Free and PNS to parallelism hardware. In addition, even when we apply it to a single core implementation, computing $\arg \min$ is faster than dividing by the sums.

\end{document}